\documentclass[11pt,a4paper]{article}

\usepackage[utf8]{inputenc}
\usepackage[T1]{fontenc}
\usepackage{lmodern}
\usepackage{amsmath,amssymb,amsfonts,amsthm}
\usepackage{graphicx}
\usepackage{booktabs}
\usepackage{longtable}
\usepackage{microtype}
\usepackage{hyperref}
\usepackage[numbers,sort&compress]{natbib}
\usepackage{geometry}
\usepackage{xcolor}
\usepackage{caption}
\usepackage{subcaption}

\hypersetup{
  colorlinks=true,
  linkcolor=blue!70!black,
  citecolor=blue!70!black,
  urlcolor=blue!70!black,
  pdfauthor={Pavan B. Govindaraju},
  pdftitle={On the Non-Negativity of Longitudinal Velocity Autocorrelations in Homogeneous Isotropic Turbulence: Kinematic Constraints, Convexity, and Spectral Broadening}
}

\newtheorem{theorem}{Theorem}

\newtheorem{proposition}[theorem]{Proposition}
\newtheorem{corollary}[theorem]{Corollary}
\theoremstyle{definition}

\title{\textbf{\Large On the Non-Negativity of Longitudinal Velocity Autocorrelations in Homogeneous Isotropic Turbulence: Kinematic Constraints, Convexity, and Spectral Broadening}}

\author{
  \textbf{Pavan B. Govindaraju}\thanks{Email: \href{mailto:pavan@splitfxm.com}{pavan@splitfxm.com}} \\
  \textit{SplitFXM, Hyderabad, India}
}

\date{\today}

\begin{document}

\maketitle

\begin{abstract}
A fundamental question in the statistical theory of homogeneous isotropic turbulence is whether the longitudinal velocity autocorrelation $f(r)$ remains non-negative and if it can be shown from first principles such as kinematic realizability. While incompressibility forces the transverse correlation $g(r)$ to exhibit a negative loop ($\int_0^\infty r g(r)\,dr = 0$), the sign behavior of $f(r)$ is subtle. We show that positive-definiteness via Bochner's and Schoenberg's theorems structurally fails to enforce pointwise positivity because the 3D spherical projection kernels are sign-changing. By mapping the problem to the one-dimensional spectrum $E_{11}(k_1)$ via the cosine transform and applying P\'olya's criterion, convexity of $E_{11}(k_1)$ provides a rigorous sufficient condition for $f(r) \ge 0$. While pure Kolmogorov $-5/3$ scaling is strictly convex, physically realizable spectra with Batchelor ($k^4$) or Saffman ($k^2$) infrared scaling are necessarily concave near $k_1 = 0$. For mature broadband spectra such as the von K\'arm\'an--Pao model, high-precision quadrature confirms that the convex inertial bulk dominates, yielding strictly positive tails and demonstrating that apparent negative dips are discrete Fourier artifacts. Also, we construct a smooth, divergence-free, finite-energy narrowband initial field that produces a genuine negative loop ($\min_r f(r) \approx -0.083$), serving as a counterexample to universal non-negativity. Spectral bandwidth is shown as the governing physical parameter, with nonlinear triad interactions rapidly broadening narrowband fields over an eddy turnover time, systematically suppressing negative excursions and preserving $f(r) \ge 0$ in mature turbulence.
\end{abstract}

\vspace{0.5em}

\section{Introduction}
\label{sec:intro}
The statistical description of homogeneous isotropic turbulence (HIT), initiated by Taylor \cite{Taylor1935} and formalized by von K\'arm\'an \& Howarth \cite{vonKarmanHowarth1938} and Robertson \cite{Robertson1940}, rests upon the two-point velocity correlation tensor
\begin{equation}
R_{ij}(\mathbf{r}, t) = \langle u_i(\mathbf{x}, t) u_j(\mathbf{x} + \mathbf{r}, t) \rangle,
\label{eq:Rij_def}
\end{equation}
where $\mathbf{u}(\mathbf{x}, t)$ is the divergence-free fluctuating velocity field and angle brackets denote an ensemble average. Under spatial isotropy and reflectional invariance, the nine components of $R_{ij}(\mathbf{r}, t)$ collapse entirely onto two scalar correlation coefficients: the longitudinal autocorrelation function $f(r, t)$ and the lateral (transverse) autocorrelation function $g(r, t)$, defined by
\begin{equation}
R_{ij}(\mathbf{r}, t) = u'^2 \left[ g(r, t)\,\delta_{ij} + \big(f(r, t) - g(r, t)\big)\,\frac{r_i r_j}{r^2} \right],
\label{eq:Rij_isotropic}
\end{equation}
with $r = |\mathbf{r}|$, $u'(t) = \langle u_1^2 \rangle^{1/2}$, and normalized such that $f(0, t) = g(0, t) = 1$.

Incompressible mass conservation ($\nabla \cdot \mathbf{u} = 0$) enforces the well-known differential kinematic relation
\begin{equation}
g(r, t) = f(r, t) + \frac{r}{2} \frac{\partial f(r, t)}{\partial r} = \frac{1}{2r} \frac{\partial}{\partial r}\big(r^2 f(r, t)\big).
\label{eq:continuity_fg}
\end{equation}
Integrating Eq.~\eqref{eq:continuity_fg} over all separation distances yields an exact integral constraint for the lateral correlation:
\begin{equation}
\int_0^\infty r\,g(r, t)\,dr = \frac{1}{2} \Big[ r^2 f(r, t) \Big]_0^\infty = 0,
\label{eq:g_zero_integral}
\end{equation}
provided that $r^2 f(r, t) \to 0$ as $r \to \infty$. Because $g(0, t) = 1 > 0$, Eq.~\eqref{eq:g_zero_integral} strictly forces $g(r, t)$ to cross zero and attain negative values over a finite interval of separations, a phenomenon ubiquitously referred to as the \emph{transverse negative loop} \cite{Batchelor1953, MoninYaglom1975, Pope2000, Davidson2004}.

In contrast, no such kinematic integral constraint exists for the longitudinal correlation $f(r, t)$. The longitudinal integral scale,
\begin{equation}
L(t) = \int_0^\infty f(r, t)\,dr,
\label{eq:integral_scale}
\end{equation}
is strictly positive for all physically realizable turbulent fields. In classical self-preservation theories, Sedov \cite{Sedov1944} demonstrated that under the hypothesis of self-similar decay ($f(r, t) = f(\xi)$ with similarity variable $\xi = r/\ell(t)$), the von K\'arm\'an--Howarth equation yields an exact closed-form solution in terms of the confluent hypergeometric function of the first kind \cite{Sedov1944, Ran2009}:
\begin{equation}
f(\xi) = {}_1F_1\left(a;\, b;\, -\frac{\xi^2}{8}\right), \qquad (b > a > 0).
\label{eq:sedov_hypergeometric}
\end{equation}
From Kummer's integral representation,
\begin{equation}
{}_1F_1(a;\, b;\, z) = \frac{\Gamma(b)}{\Gamma(a)\Gamma(b-a)} \int_0^1 e^{z t}\, t^{a-1} (1-t)^{b-a-1}\,dt,
\label{eq:kummer_integral}
\end{equation}
it is straightforward to prove that for any $z \le 0$ (here $z = -\xi^2/8 \le 0$) with $b > a > 0$, the integrand is strictly positive on $(0, 1)$, proving that Sedov's self-similar solutions unconditionally satisfy $f(\xi) > 0$ for all $\xi \in [0, \infty)$ (shown in Figure~\ref{fig:sedov_solution}).

\begin{figure}[t]
\centering
\includegraphics[width=0.68\textwidth]{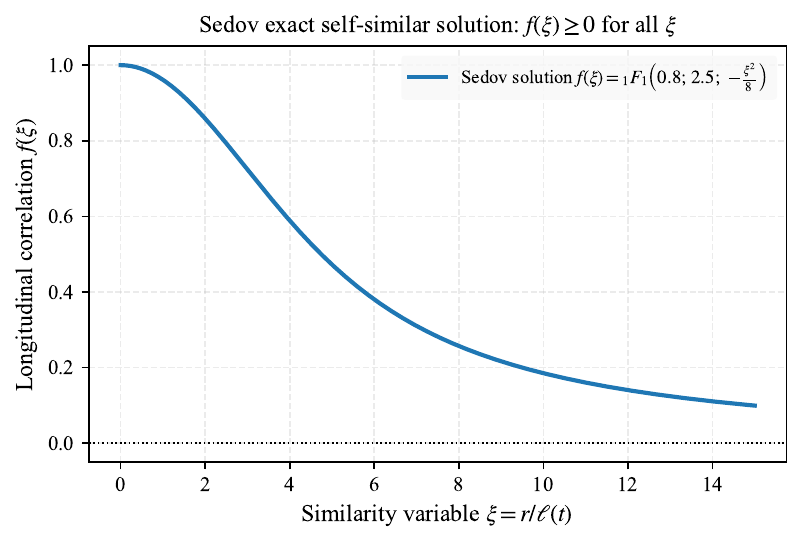}
\caption{The exact self-similar solution of the von K\'arm\'an--Howarth equation derived by Sedov \cite{Sedov1944}, $f(\xi) = {}_1F_1(0.8;\, 2.5;\, -\xi^2/8)$. Because Kummer's integral representation maintains a strictly positive integrand for negative arguments, self-preserving solutions are unconditionally non-negative ($f(\xi) \ge 0$).}
\label{fig:sedov_solution}
\end{figure}

However, physical and synthetic turbulent fields are not constrained \emph{a priori} to self-preserving decay states. In textbooks and classical closure theories, $f(r)$ is nevertheless frequently depicted as a strictly non-negative function \cite{Pope2000, Frisch1996}. This raises a foundational question:
\begin{quote}
\emph{Is the longitudinal velocity autocorrelation function $f(r)$ non-negative for all spatial separations $r$ from first principles of kinematic realizability, or can $f(r)$ become negative outside the self-similar class?}
\end{quote}

This question has persisted across both theoretical mechanics and experimental measurements. In experimental grid turbulence, while the lateral correlation $g(r)$ ubiquitously crosses zero to satisfy incompressibility, measured longitudinal correlations $f(r)$ in mature decay regions are consistently observed to remain non-negative within experimental uncertainty \cite{StewartTownsend1951, ComteBellotCorrsin1971, FrenkielKlebanoffHuang1979}. In their comprehensive space-time correlation measurements, Comte-Bellot \& Corrsin \cite{ComteBellotCorrsin1971} explicitly documented that full-band longitudinal correlations following the mean flow do not become negative. Furthermore, they demonstrated that apparent negative excursions in single-probe temporal autocorrelations are typically artifacts of electronic AC-coupling or high-pass filtering (which removes the zero-wavenumber energy and enforces a zero-integral constraint), whereas isolating narrow spectral bands electronically produces pronounced oscillatory negative loops. Yet, from a fundamental theoretical perspective, it has remained unsettled whether the non-negativity $f(r) \ge 0$ is a strict consequence of kinematic realizability and Navier-Stokes dynamics, or whether kinematically valid initial velocity fields can genuinely attain negative values.

In this work, we provide a theoretical and numerical basis to this question:
\begin{enumerate}
\item We show that positive definiteness in $\mathbb{R}^3$ via Bochner's and Schoenberg's theorems \cite{Bochner1959, Schoenberg1938} cannot establish $f(r) \ge 0$ because the associated projection kernels oscillate and change sign.
\item By analyzing the one-dimensional longitudinal energy spectrum $E_{11}(k_1)$, we prove that $E_{11}(k_1)$ is \emph{unconditionally monotonically decreasing} for every realizable 3D spectrum $E(k) \ge 0$, and we establish via P\'olya's criterion \cite{Polya1949} that convexity of $E_{11}(k_1)$ is a rigorous sufficient condition for $f(r) \ge 0$.
\item We prove that physically admissible turbulent spectra, encompassing both the Batchelor ($E(k) \sim k^4$) and Saffman ($E(k) \sim k^2$) infrared universality classes, are necessarily concave in a finite neighborhood of the origin ($k_1 \to 0$), so P\'olya's sufficient condition is never strictly satisfied in a global sense.
\item For mature broadband spectra (e.g., von K\'arm\'an--Pao), high-precision quadrature demonstrates that the convex inertial bulk dominates the cosine transform, ensuring $f(r) \ge 0$ and identifying previous negative dips as discrete Fourier transform (FFT) aliasing artifacts.
\item We construct an admissible, smooth, divergence-free, finite-energy narrowband spectrum satisfying Batchelor's infrared constraint that yields a pronounced, provable negative loop ($\min_r f(r) \approx -0.083$). This demonstrates that non-self-similar, narrowband velocity fields do not fall within Sedov's non-negative class, proving that $f(r) \ge 0$ is not a universal kinematic law. This state represents a legitimate Navier-Stokes initial condition whose evolution is guaranteed by classical local existence theorems \cite{Leray1934, FujitaKato1964}, revealing that spectral broadbandness generated by nonlinear triad interactions is the fundamental physical mechanism that maintains $f(r) \ge 0$ in mature turbulent flows.
\end{enumerate}

\section{Kinematics and the Failure of Naive Positive Definiteness}
\label{sec:bochner_failure}

In this section, we formulate the kinematics of homogeneous isotropic turbulence from first principles, derive the isotropic correlation and spectral tensors, and demonstrate why classical multidimensional positive-definiteness arguments (via Bochner's and Schoenberg's theorems) structurally fail to enforce non-negativity of the longitudinal correlation $f(r)$.

\subsection{Invariant Tensor Theory and Derivation of the Isotropic Form of \texorpdfstring{$R_{ij}(\mathbf{r})$}{Rij(r)}}
\label{subsec:rij_derivation}

We begin by establishing from first principles why the two-point velocity correlation tensor $R_{ij}(\mathbf{r}) = \langle u_i(\mathbf{x}) u_j(\mathbf{x}+\mathbf{r}) \rangle$ in homogeneous isotropic turbulence depends entirely on just two scalar correlation functions, $f(r)$ and $g(r)$.

\paragraph{1. Homogeneity (Translational Invariance):}
Statistical homogeneity implies that the correlation depends only on the separation vector $\mathbf{r} = \mathbf{x}' - \mathbf{x}$ and is independent of the absolute spatial location $\mathbf{x}$. Evaluating the function at $-\mathbf{r}$ yields
\begin{equation}
R_{ij}(-\mathbf{r}) = \langle u_i(\mathbf{x}) u_j(\mathbf{x}-\mathbf{r}) \rangle = \langle u_i(\mathbf{x}+\mathbf{r}) u_j(\mathbf{x}) \rangle = \langle u_j(\mathbf{x}) u_i(\mathbf{x}+\mathbf{r}) \rangle = R_{ji}(\mathbf{r}).
\label{eq:Rij_homogeneity_sym}
\end{equation}

\paragraph{2. Isotropy (Rotational Invariance):}
Statistical isotropy requires that under any coordinate rotation $Q \in \mathrm{SO}(3)$ (where $Q_{ik} Q_{jk} = \delta_{ij}$ and $\det Q = +1$), the correlation tensor transforms as an isotropic second-order tensor field:
\begin{equation}
R_{ij}(Q\mathbf{r}) = Q_{im} Q_{jn} R_{mn}(\mathbf{r}).
\label{eq:Rij_rotation}
\end{equation}
By the invariant tensor representation theorem for vector-valued isotropic functions (\cite{Robertson1940}, \cite{Batchelor1953}), any second-order tensor constructed from a single vector argument $\mathbf{r}$ and the fundamental isotropic tensors ($\delta_{ij}$ and the Levi-Civita permutation symbol $\epsilon_{ijk}$) must be a linear combination of the three fundamental isotropic second-order forms:
\begin{equation}
R_{ij}(\mathbf{r}) = A(r)\,\delta_{ij} + B(r)\,r_i r_j + C(r)\,\epsilon_{ijk} r_k,
\label{eq:Rij_general_isotropic}
\end{equation}
where $r = |\mathbf{r}| = \sqrt{r_k r_k}$ is the sole scalar invariant, and $A(r), B(r), C(r)$ are scalar functions of $r$.

\paragraph{3. Reflection Invariance (Parity / Zero Mean Helicity):}
Under a spatial reflection through the origin ($\mathbf{x} \to -\mathbf{x}$), a true (polar) velocity vector transforms as $\mathbf{u}(\mathbf{x}) \to -\mathbf{u}(-\mathbf{x})$. Statistical reflection invariance (parity invariance) requires that the two-point velocity statistics remain invariant under this transformation:
\begin{align}
R_{ij}(\mathbf{r}) &= \langle u_i(\mathbf{x}) u_j(\mathbf{x}+\mathbf{r}) \rangle = \langle [-u_i(-\mathbf{x})] [-u_j(-\mathbf{x}-\mathbf{r})] \rangle \nonumber \\
&= \langle u_i(-\mathbf{x}) u_j(-\mathbf{x}-\mathbf{r}) \rangle = R_{ij}(-\mathbf{r}).
\label{eq:Rij_reflection}
\end{align}
Evaluating Eq.~\eqref{eq:Rij_general_isotropic} at $-\mathbf{r}$ yields:
\begin{align}
R_{ij}(-\mathbf{r}) &= A(r)\,\delta_{ij} + B(r)\,(-r_i)(-r_j) + C(r)\,\epsilon_{ijk}(-r_k) \nonumber \\
&= A(r)\,\delta_{ij} + B(r)\,r_i r_j - C(r)\,\epsilon_{ijk} r_k.
\end{align}
Equating $R_{ij}(\mathbf{r}) = R_{ij}(-\mathbf{r})$ from Eq.~\eqref{eq:Rij_reflection} gives:
\begin{equation}
R_{ij}(\mathbf{r}) - R_{ij}(-\mathbf{r}) = 2C(r)\,\epsilon_{ijk} r_k = 0 \implies C(r) \equiv 0.
\end{equation}
Equivalently, combining reflection invariance ($R_{ij}(\mathbf{r}) = R_{ij}(-\mathbf{r})$) with homogeneity ($R_{ij}(-\mathbf{r}) = R_{ji}(\mathbf{r})$ from Eq.~\eqref{eq:Rij_homogeneity_sym}) requires index symmetry $R_{ij}(\mathbf{r}) = R_{ji}(\mathbf{r})$. Because the Levi-Civita term is strictly antisymmetric under index transposition ($\epsilon_{jik} r_k = -\epsilon_{ijk} r_k$), its presence violates reflection symmetry unless $C(r) \equiv 0$.

Physically, $C(r)$ characterizes turbulent helicity ($\langle \mathbf{u} \cdot (\nabla \times \mathbf{u}) \rangle$), a pseudoscalar that identically vanishes in mirror-symmetric isotropic turbulence. Hence, $R_{ij}(\mathbf{r})$ is strictly symmetric:
\begin{equation}
R_{ij}(\mathbf{r}) = A(r)\,\delta_{ij} + B(r)\,r_i r_j.
\label{eq:Rij_symmetric}
\end{equation}

\paragraph{4. Identification with Longitudinal and Lateral Correlations:}
To determine $A(r)$ and $B(r)$ physically, align the coordinate system such that the separation vector lies along the $x_1$-axis: $\mathbf{r} = (r, 0, 0)$.
\begin{itemize}
\item The \emph{longitudinal autocorrelation} is the correlation of the velocity component parallel to the separation vector:
\begin{equation}
R_{11}(r \mathbf{e}_1) = \langle u_1(\mathbf{x}) u_1(\mathbf{x} + r\mathbf{e}_1) \rangle \equiv u'^2 f(r).
\label{eq:f_def_parallel}
\end{equation}
Evaluating Eq.~\eqref{eq:Rij_symmetric} for $i=j=1$ gives
\begin{equation}
R_{11}(r \mathbf{e}_1) = A(r) + B(r)\,r^2 = u'^2 f(r).
\label{eq:match_longitudinal}
\end{equation}
\item The \emph{lateral (transverse) autocorrelation} is the correlation of a velocity component perpendicular to the separation vector (e.g., $u_2$ or $u_3$):
\begin{equation}
R_{22}(r \mathbf{e}_1) = \langle u_2(\mathbf{x}) u_2(\mathbf{x} + r\mathbf{e}_1) \rangle \equiv u'^2 g(r).
\label{eq:g_def_transverse}
\end{equation}
Evaluating Eq.~\eqref{eq:Rij_symmetric} for $i=j=2$ gives
\begin{equation}
R_{22}(r \mathbf{e}_1) = A(r) = u'^2 g(r).
\label{eq:match_lateral}
\end{equation}
\end{itemize}
Substituting $A(r) = u'^2 g(r)$ into Eq.~\eqref{eq:match_longitudinal} yields
\begin{equation}
B(r)\,r^2 = u'^2 \big(f(r) - g(r)\big) \implies B(r) = u'^2\,\frac{f(r) - g(r)}{r^2}.
\label{eq:solve_B}
\end{equation}
Substituting $A(r)$ and $B(r)$ back into Eq.~\eqref{eq:Rij_symmetric} gives the canonical tensorial form of von K\'arm\'an \& Howarth \cite{vonKarmanHowarth1938}:
\begin{equation}
R_{ij}(\mathbf{r}) = u'^2 \left[ g(r)\,\delta_{ij} + \big(f(r) - g(r)\big)\,\frac{r_i r_j}{r^2} \right].
\label{eq:Rij_canonical_proven}
\end{equation}
Thus, any homogeneous, isotropic, reflection-symmetric second-order velocity correlation tensor in three dimensions is uniquely and completely determined by the two scalar functions $f(r)$ and $g(r)$.

\paragraph{5. Incompressible Continuity and Derivation of the Transverse Zero-Crossing:}
We now prove from first principles why the lateral correlation $g(r)$ is kinematically compelled to cross zero and attain negative values.

Mass conservation for an incompressible fluid requires the velocity field to be divergence-free, $\nabla \cdot \mathbf{u} = \frac{\partial u_j}{\partial x_j} = 0$. Taking the divergence of the correlation tensor with respect to the separation coordinates $r_j$ yields
\begin{equation}
\frac{\partial R_{ij}(\mathbf{r})}{\partial r_j} = \left\langle u_i(\mathbf{x}) \frac{\partial u_j(\mathbf{x}+\mathbf{r})}{\partial r_j} \right\rangle = 0.
\label{eq:div_Rij_zero}
\end{equation}
Differentiating Eq.~\eqref{eq:Rij_canonical_proven} directly using $\frac{\partial r}{\partial r_j} = \frac{r_j}{r}$ and $\frac{\partial}{\partial r_j}(r_i r_j) = \delta_{ij}r_j + r_i\delta_{jj} = 4r_i$:
\begin{align}
\frac{1}{u'^2}\frac{\partial R_{ij}(\mathbf{r})}{\partial r_j} &= g'(r)\frac{r_j}{r}\delta_{ij} + \frac{d}{dr}\left(\frac{f(r)-g(r)}{r^2}\right)\frac{r_j}{r} r_i r_j + \frac{f(r)-g(r)}{r^2}\,(4r_i) \nonumber\\
&= g'(r)\frac{r_i}{r} + \left(\frac{f'(r)-g'(r)}{r^2} - \frac{2(f(r)-g(r))}{r^3}\right)\frac{r^2 r_i}{r} + \frac{4(f(r)-g(r)) r_i}{r^2} \nonumber\\
&= \frac{r_i}{r} \left[ f'(r) + \frac{2\big(f(r)-g(r)\big)}{r} \right] = 0.
\label{eq:div_derivation_steps}
\end{align}
Since Eq.~\eqref{eq:div_derivation_steps} must vanish for any arbitrary separation vector $r_i \neq 0$, the bracketed scalar term must vanish identically:
\begin{equation}
f'(r) + \frac{2\big(f(r)-g(r)\big)}{r} = 0 \implies g(r) = f(r) + \frac{r}{2}\frac{df(r)}{dr}.
\label{eq:continuity_fg_derived}
\end{equation}
Multiplying Eq.~\eqref{eq:continuity_fg_derived} by $2r$ reveals that the right-hand side is an exact total derivative:
\begin{equation}
2r\,g(r) = 2r f(r) + r^2 f'(r) = \frac{d}{dr}\big(r^2 f(r)\big).
\label{eq:total_derivative_rg}
\end{equation}
Integrating Eq.~\eqref{eq:total_derivative_rg} over all spatial separations $r \in [0, \infty)$ yields
\begin{equation}
\int_0^\infty r\,g(r)\,dr = \frac{1}{2} \int_0^\infty \frac{d}{dr}\big(r^2 f(r)\big)\,dr = \frac{1}{2}\Big[ r^2 f(r) \Big]_0^\infty.
\label{eq:integral_rg_derivation}
\end{equation}
For any physical turbulent field with finite energy and finite integral scale, the correlation decays asymptotically such that $r^2 f(r) \to 0$ as $r \to \infty$. At the lower limit, $\lim_{r\to 0} r^2 f(r) = 0 \cdot 1 = 0$. Consequently, the boundary terms vanish identically, establishing the exact integral constraint:
\begin{equation}
\int_0^\infty r\,g(r)\,dr = 0.
\label{eq:rg_zero_constraint_proven}
\end{equation}
Because the normalization requires $g(0) = 1 > 0$ and $g(r)$ is continuous, $g(r) > 0$ in an initial neighborhood $[0, r_0)$. On this interval, the integrand $r g(r)$ is strictly positive, yielding a strictly positive area $\int_0^{r_0} r g(r)\,dr > 0$. For the entire integral on $[0, \infty)$ to balance to identically zero, the integrand $r g(r)$ (and hence $g(r)$, since $r > 0$) \textbf{must strictly cross zero and attain negative values} over a finite range of separations ($r \in (r_1, r_2)$). This completes the kinematic proof that the transverse correlation $g(r)$ must exhibit a negative loop.

\subsection{Spectral Tensors and Kinematic Realizability}
\label{subsec:spectral_tensors}

In statistical fluid mechanics, the energy distribution across spatial scales is analyzed through the Fourier representation of the fluctuating velocity field $\mathbf{u}(\mathbf{x})$. The two-point velocity correlation tensor $R_{ij}(\mathbf{r}) = \langle u_i(\mathbf{x}) u_j(\mathbf{x}+\mathbf{r}) \rangle$ and the velocity spectrum tensor $\Phi_{ij}(\mathbf{k})$ form a three-dimensional Fourier transform pair:
\begin{equation}
R_{ij}(\mathbf{r}) = \iiint_{\mathbb{R}^3} \Phi_{ij}(\mathbf{k})\,e^{i \mathbf{k} \cdot \mathbf{r}}\,d\mathbf{k}, \qquad
\Phi_{ij}(\mathbf{k}) = \frac{1}{(2\pi)^3} \iiint_{\mathbb{R}^3} R_{ij}(\mathbf{r})\,e^{-i \mathbf{k} \cdot \mathbf{r}}\,d\mathbf{r},
\label{eq:spectral_tensor_pair}
\end{equation}
where $\mathbf{k} = (k_1, k_2, k_3)$ is the wavevector and $k = |\mathbf{k}| = (k_i k_i)^{1/2}$ is the wavenumber magnitude.

\subsubsection{Derivation of the Isotropic Spectrum Tensor from First Principles}
We now derive the unique algebraic structure of $\Phi_{ij}(\mathbf{k})$ from first principles under the physical symmetries of homogeneous, isotropic, reflection-symmetric, incompressible turbulence:

\paragraph{1. General Isotropic Form in Wavevector Space:}
Under a coordinate rotation $Q \in \mathrm{SO}(3)$ (with $Q_{ik} Q_{jk} = \delta_{ij}$ and $\det Q = +1$), statistical isotropy requires the spectrum tensor to transform as an isotropic second-order tensor field:
\begin{equation}
\Phi_{ij}(Q\mathbf{k}) = Q_{im} Q_{jn} \Phi_{mn}(\mathbf{k}).
\label{eq:Phi_rotation}
\end{equation}
By the invariant tensor representation theorem for vector-valued isotropic functions (\cite{Robertson1940}, \cite{Batchelor1953}), any second-order tensor depending solely on a single vector argument $\mathbf{k}$ and the isotropic tensors ($\delta_{ij}$, $\epsilon_{ijk}$) must be expressed in terms of the scalar invariant $k = |\mathbf{k}| = \sqrt{k_m k_m}$ as:
\begin{equation}
\Phi_{ij}(\mathbf{k}) = A_s(k)\,\delta_{ij} + B_s(k)\,\frac{k_i k_j}{k^2} + C_s(k)\,\epsilon_{ijp} \frac{k_p}{k},
\label{eq:Phi_general}
\end{equation}
where $A_s(k)$, $B_s(k)$, and $C_s(k)$ are scalar functions of the wavenumber magnitude $k$ (with normalization factors $k^2$ and $k$ chosen so that all three coefficient functions share identical physical dimensions).

\paragraph{2. Reflection Symmetry (Parity / Zero Mean Helicity):}
Under a spatial reflection through the origin ($\mathbf{x} \to -\mathbf{x}$), the polar velocity vector transforms as $\mathbf{u}(\mathbf{x}) \to \mathbf{u}'(\mathbf{x}) = -\mathbf{u}(-\mathbf{x})$. Taking the continuous spatial Fourier transform, the spectral velocity amplitude transforms under parity as:
\begin{align}
\hat{\mathbf{u}}'(\mathbf{k}) &= \frac{1}{(2\pi)^3} \iiint_{\mathbb{R}^3} [-\mathbf{u}(-\mathbf{x})]\,e^{-i \mathbf{k} \cdot \mathbf{x}}\,d\mathbf{x} \nonumber \\
&= -\frac{1}{(2\pi)^3} \iiint_{\mathbb{R}^3} \mathbf{u}(\mathbf{y})\,e^{-i(-\mathbf{k})\cdot\mathbf{y}}\,d\mathbf{y} \nonumber \\
&= -\hat{\mathbf{u}}(-\mathbf{k}),
\label{eq:u_hat_parity}
\end{align}
via the substitution $\mathbf{y} = -\mathbf{x}$ ($d\mathbf{x} = d\mathbf{y}$). In physical space, statistical reflection invariance requires the two-point velocity correlation tensor to satisfy $R_{ij}(\mathbf{r}) = R_{ij}(-\mathbf{r})$ (Eq.~\eqref{eq:Rij_reflection}). Applying this directly to the Fourier transform definition of $\Phi_{ij}(\mathbf{k})$ in Eq.~\eqref{eq:spectral_tensor_pair} and substituting $\mathbf{r}' = -\mathbf{r}$ ($d\mathbf{r} = d\mathbf{r}'$):
\begin{align}
\Phi_{ij}(\mathbf{k}) &= \frac{1}{(2\pi)^3} \iiint_{\mathbb{R}^3} R_{ij}(\mathbf{r})\,e^{-i \mathbf{k} \cdot \mathbf{r}}\,d\mathbf{r} \nonumber \\
&= \frac{1}{(2\pi)^3} \iiint_{\mathbb{R}^3} R_{ij}(-\mathbf{r})\,e^{-i \mathbf{k} \cdot \mathbf{r}}\,d\mathbf{r} \nonumber \\
&= \frac{1}{(2\pi)^3} \iiint_{\mathbb{R}^3} R_{ij}(\mathbf{r}')\,e^{-i(-\mathbf{k})\cdot\mathbf{r}'}\,d\mathbf{r}' \nonumber \\
&= \Phi_{ij}(-\mathbf{k}).
\label{eq:Phi_reflection}
\end{align}
Equivalently, in spectral space, the two-point spectral correlation satisfies:
\begin{equation}
\langle \hat{u}'_i(\mathbf{k})\,\hat{u}'^*_j(\mathbf{k}') \rangle = \langle [-\hat{u}_i(-\mathbf{k})]\,[-\hat{u}^*_j(-\mathbf{k}')] \rangle = \langle \hat{u}_i(-\mathbf{k})\,\hat{u}^*_j(-\mathbf{k}') \rangle,
\end{equation}
confirming the exact parity invariance $\Phi_{ij}(\mathbf{k}) = \Phi_{ij}(-\mathbf{k})$.

Evaluating Eq.~\eqref{eq:Phi_general} at $-\mathbf{k}$ (noting that the scalar wavenumber magnitude $|-\mathbf{k}| = k$ is strictly invariant):
\begin{align}
\Phi_{ij}(-\mathbf{k}) &= A_s(k)\,\delta_{ij} + B_s(k)\,\frac{(-k_i)(-k_j)}{k^2} + C_s(k)\,\epsilon_{ijp} \frac{-k_p}{k} \nonumber \\
&= A_s(k)\,\delta_{ij} + B_s(k)\,\frac{k_i k_j}{k^2} - C_s(k)\,\epsilon_{ijp} \frac{k_p}{k}.
\end{align}
Equating $\Phi_{ij}(\mathbf{k}) = \Phi_{ij}(-\mathbf{k})$ from Eq.~\eqref{eq:Phi_reflection} requires:
\begin{equation}
\Phi_{ij}(\mathbf{k}) - \Phi_{ij}(-\mathbf{k}) = 2C_s(k)\,\epsilon_{ijp}\frac{k_p}{k} = 0 \implies C_s(k) \equiv 0.
\end{equation}
Thus, reflection symmetry eliminates the pseudo-tensor (helical) contribution, guaranteeing that $\Phi_{ij}(\mathbf{k})$ is strictly symmetric in its indices ($\Phi_{ij}(\mathbf{k}) = \Phi_{ji}(\mathbf{k})$):
\begin{equation}
\Phi_{ij}(\mathbf{k}) = A_s(k)\,\delta_{ij} + B_s(k)\,\frac{k_i k_j}{k^2}.
\label{eq:Phi_symmetric}
\end{equation}

\paragraph{3. Incompressibility / Divergence-Free Constraint:}
Mass conservation in physical space for an incompressible fluid requires $\partial u_j / \partial x_j = 0$, which in terms of the two-point correlation tensor translates to $\partial R_{ij}(\mathbf{r}) / \partial r_j = 0$. Differentiating the inverse Fourier transform~\eqref{eq:spectral_tensor_pair}:
\begin{equation}
\frac{\partial R_{ij}(\mathbf{r})}{\partial r_j} = \frac{\partial}{\partial r_j} \iiint_{\mathbb{R}^3} \Phi_{ij}(\mathbf{k})\,e^{i \mathbf{k} \cdot \mathbf{r}}\,d\mathbf{k} = \iiint_{\mathbb{R}^3} i k_j \Phi_{ij}(\mathbf{k})\,e^{i \mathbf{k} \cdot \mathbf{r}}\,d\mathbf{k} = 0.
\end{equation}
Since this vanishing integral must hold for all separation vectors $\mathbf{r} \in \mathbb{R}^3$, the Fourier integrand must satisfy the exact orthogonality constraint:
\begin{equation}
k_j \Phi_{ij}(\mathbf{k}) = 0 \quad \text{for all } \mathbf{k} \neq \mathbf{0}.
\label{eq:fourier_continuity}
\end{equation}
Contracting the symmetric form~\eqref{eq:Phi_symmetric} with $k_j$:
\begin{equation}
k_j \Phi_{ij}(\mathbf{k}) = k_j \left( A_s(k)\,\delta_{ij} + B_s(k)\,\frac{k_i k_j}{k^2} \right) = A_s(k)\,k_i + B_s(k)\,\frac{k_i (k_j k_j)}{k^2} = \big[ A_s(k) + B_s(k) \big] k_i = 0.
\end{equation}
Because $k_i \neq 0$ for non-zero wavevectors, we obtain the exact algebraic link:
\begin{equation}
B_s(k) = -A_s(k).
\label{eq:Bs_equals_minus_As}
\end{equation}
Substituting Eq.~\eqref{eq:Bs_equals_minus_As} into Eq.~\eqref{eq:Phi_symmetric} yields the transverse projection tensor form:
\begin{equation}
\Phi_{ij}(\mathbf{k}) = A_s(k) \left( \delta_{ij} - \frac{k_i k_j}{k^2} \right) = \frac{A_s(k)}{k^2} \left( k^2 \delta_{ij} - k_i k_j \right),
\label{eq:Phi_projected}
\end{equation}
where $P_{ij}(\mathbf{k}) = \delta_{ij} - k_i k_j / k^2$ is the classical orthogonal projection operator onto the solenoidal (divergence-free) subspace perpendicular to the wavevector $\mathbf{k}$.

\paragraph{4. Connection to the 3D Turbulent Kinetic Energy Spectrum $E(k)$:}
The total turbulent kinetic energy per unit mass is defined by the trace of the zero-separation correlation tensor, $R_{ii}(\mathbf{0}) = \langle u_i(\mathbf{x}) u_i(\mathbf{x}) \rangle = 3 u'^2$. Evaluating the inverse Fourier transform~\eqref{eq:spectral_tensor_pair} at $\mathbf{r} = \mathbf{0}$:
\begin{equation}
\frac{3}{2} u'^2 = \frac{1}{2} R_{ii}(\mathbf{0}) = \frac{1}{2} \iiint_{\mathbb{R}^3} \Phi_{ii}(\mathbf{k})\,d\mathbf{k}.
\label{eq:TKE_trace_integral}
\end{equation}
Contracting the indices of the projected tensor~\eqref{eq:Phi_projected} yields
\begin{equation}
\Phi_{ii}(\mathbf{k}) = A_s(k) \left( \delta_{ii} - \frac{k_i k_i}{k^2} \right) = A_s(k) (3 - 1) = 2 A_s(k).
\end{equation}
Transforming the volume element $d\mathbf{k}$ to spherical polar coordinates ($d\mathbf{k} = k^2 \sin\theta\,dk\,d\theta\,d\phi$), the angular integration over the spherical shell $S^2$ of radius $k$ yields:
\begin{equation}
\frac{1}{2} \iiint_{\mathbb{R}^3} \Phi_{ii}(\mathbf{k})\,d\mathbf{k} = \frac{1}{2} \int_0^\infty 2 A_s(k) \left( \int_0^\pi \sin\theta\,d\theta \int_0^{2\pi} d\phi \right) k^2\,dk = \int_0^\infty 4\pi k^2 A_s(k)\,dk.
\label{eq:spherical_shell_int}
\end{equation}
By physical definition, the 3D energy spectrum $E(k)$ is the kinetic energy density per unit wavenumber magnitude, such that the total turbulent kinetic energy satisfies $\frac{3}{2} u'^2 = \int_0^\infty E(k)\,dk$. Equating integrands on $[0, \infty)$:
\begin{equation}
E(k) = 4\pi k^2 A_s(k) \implies A_s(k) = \frac{E(k)}{4\pi k^2}.
\label{eq:As_relation}
\end{equation}
Substituting Eq.~\eqref{eq:As_relation} into Eq.~\eqref{eq:Phi_projected} establishes the canonical Batchelor form of the isotropic velocity spectrum tensor:
\begin{equation}
\Phi_{ij}(\mathbf{k}) = \frac{E(k)}{4\pi k^4} \left( k^2 \delta_{ij} - k_i k_j \right).
\label{eq:Phi_ij_isotropic}
\end{equation}

\subsubsection{Role of \texorpdfstring{$\Phi_{ij}(\mathbf{k})$}{Phi\_ij(k)} in the Analysis}
The spectrum tensor $\Phi_{ij}(\mathbf{k})$ is the central mathematical bridge throughout this paper for three fundamental reasons:
\begin{enumerate}
\item \textbf{Exact Tensor Transformations to Physical Space:} It provides the exact closed-form link between the 3D spectrum $E(k)$ and physical correlations. Specifically, spherical integration of $\Phi_{ii}(\mathbf{k}) e^{i\mathbf{k}\cdot\mathbf{r}}$ yields the Schoenberg trace relation~\eqref{eq:trace_rep}, while $\Phi_{11}(\mathbf{k}) e^{i\mathbf{k}\cdot\mathbf{r}}$ directly produces the longitudinal kernel representation~\eqref{eq:f_kernel_rep}.
\item \textbf{Derivation of the 1D Spectrum $E_{11}(k_1)$:} In Section~\ref{subsec:1d_spectrum}, the 1D longitudinal spectrum $E_{11}(k_1)$ is obtained precisely by integrating $\Phi_{11}(\mathbf{k})$ over the transverse wavenumber plane $(k_2, k_3)$, which leads directly to the fundamental integral relation $E_{11}(k_1) = \int_{k_1}^\infty \frac{E(k)}{k}\left(1 - \frac{k_1^2}{k^2}\right)dk$.
\item \textbf{Kinematic Realizability from First Principles:}
Physical realizability demands that any arbitrary linear combination of velocity components, $v(\mathbf{x}) = a_i u_i(\mathbf{x})$ ($\mathbf{a} \in \mathbb{C}^3$), possesses non-negative energy across all spatial scales. In spectral space, the velocity field decomposes into uncorrelated Fourier increments via Cram\'er's representation, $\mathbf{u}(\mathbf{x}) = \iiint_{\mathbb{R}^3} e^{i\mathbf{k}\cdot\mathbf{x}}\, d\hat{\mathbf{u}}(\mathbf{k})$, with $\langle d\hat{u}_i(\mathbf{k}) d\hat{u}_j^*(\mathbf{k}') \rangle = \Phi_{ij}(\mathbf{k})\,\delta(\mathbf{k}-\mathbf{k}')\,d\mathbf{k}\,d\mathbf{k}'$. The energy in an individual Fourier mode projected along $\mathbf{a}$ is therefore
\begin{equation}
\langle |a_i d\hat{u}_i(\mathbf{k})|^2 \rangle = \Big[ a_i^* \Phi_{ij}(\mathbf{k}) a_j \Big] d\mathbf{k} \ge 0 \quad \text{for all } \mathbf{a} \in \mathbb{C}^3, \; \mathbf{k} \in \mathbb{R}^3,
\label{eq:quadratic_form_realizability}
\end{equation}
which requires $\Phi_{ij}(\mathbf{k})$ to be positive semi-definite pointwise. Substituting the isotropic solenoidal form~\eqref{eq:Phi_ij_isotropic} into this quadratic form:
\begin{equation}
a_i^* \Phi_{ij}(\mathbf{k}) a_j = \frac{E(k)}{4\pi k^2}\, a_i^* P_{ij}(\mathbf{k}) a_j = \frac{E(k)}{4\pi k^2} \left( |\mathbf{a}|^2 - \frac{|\mathbf{k} \cdot \mathbf{a}|^2}{k^2} \right) = \frac{E(k)}{4\pi k^2}\, |\mathbf{a}_\perp|^2,
\label{eq:realizability_derivation}
\end{equation}
where $\mathbf{a}_\perp = \mathbf{P}(\mathbf{k})\mathbf{a}$ is the divergence-free projection of $\mathbf{a}$ orthogonal to $\mathbf{k}$. Because $|\mathbf{a}_\perp|^2 \ge 0$ for all $\mathbf{a}$, non-negativity of the quadratic form holds if and only if
\begin{equation}
E(k) \ge 0 \quad \text{for all } k \in [0, \infty).
\label{eq:realizability_Ek}
\end{equation}
(Equivalently, the eigenvalue problem $\Phi_{ij}(\mathbf{k}) v_j = \lambda v_i$ follows directly from the projection operator $P_{ij}(\mathbf{k}) = \delta_{ij} - k_i k_j/k^2$: contracting with the longitudinal unit vector $\mathbf{e}_3 = \mathbf{k}/k$ gives $P_{ij} k_j = 0$, yielding the longitudinal eigenvalue $\lambda_3 = 0$ (enforcing incompressibility), while contracting with any orthogonal unit vector $\mathbf{e} \perp \mathbf{k}$ gives $P_{ij} e_j = e_i$, yielding a doubly degenerate transverse eigenvalue $\lambda_{1,2} = E(k)/(4\pi k^2)$ spanning the two solenoidal polarization directions. Hence, all eigenvalues of $\Phi_{ij}(\mathbf{k})$ are non-negative if and only if $E(k) \ge 0$.)
\end{enumerate}

\subsection{Schoenberg's Theorem and Radial Projection Kernels}

A natural theoretical question is whether the non-negativity of the 3D energy spectrum, $E(k) \ge 0$, mathematically guarantees the pointwise non-negativity of the spatial correlation functions ($f(r) \ge 0$ and $g(r) \ge 0$) via classical harmonic analysis.

By Bochner's theorem \cite{Bochner1959}, a spatial correlation tensor $R_{ij}(\mathbf{r})$ is positive-definite on $\mathbb{R}^3$ if and only if it is the Fourier transform of a non-negative matrix-valued spectral measure $\Phi_{ij}(\mathbf{k}) \ge 0$. As established in Eq.~\eqref{eq:realizability_Ek}, kinematic realizability is precisely equivalent to $E(k) \ge 0$. For isotropic scalar fields, Schoenberg's representation theorem \cite{Schoenberg1938} further shows that a radial function $\psi(r)$ on $\mathbb{R}^n$ is positive-definite if and only if it admits the integral representation
\begin{equation}
\psi(r) = \int_0^\infty \Omega_n(kr)\,d\nu(k), \qquad d\nu(k) \ge 0,
\label{eq:schoenberg_rep}
\end{equation}
where the radial projection kernel $\Omega_n(x)$ is given in terms of the Bessel function of the first kind:
\begin{equation}
\Omega_n(x) = \Gamma\left(\frac{n}{2}\right) \left(\frac{2}{x}\right)^{\frac{n-2}{2}} J_{\frac{n-2}{2}}(x).
\label{eq:schoenberg_kernel_n}
\end{equation}
In three dimensions ($n=3$), the Schoenberg kernel reduces to the elementary spherical sinc function:
\begin{equation}
\Omega_3(x) = \frac{\sin x}{x}.
\label{eq:omega3_kernel}
\end{equation}

To establish the direct connection with the turbulent velocity field, we evaluate the inverse Fourier transform of the spectral trace $\Phi_{ii}(\mathbf{k}) = 2E(k)/(4\pi k^2)$ to obtain the physical correlation trace $R_{ii}(\mathbf{r}) = u'^2 [f(r) + 2g(r)]$. Substituting Eq.~\eqref{eq:Phi_ij_isotropic} into Eq.~\eqref{eq:spectral_tensor_pair} and integrating in spherical coordinates ($\mathbf{k} \cdot \mathbf{r} = kr\cos\theta$, $d\mathbf{k} = k^2\sin\theta\,dk\,d\theta\,d\phi$):
\begin{align}
u'^2 \big[ f(r) + 2g(r) \big] &= \iiint_{\mathbb{R}^3} \Phi_{ii}(\mathbf{k})\,e^{i \mathbf{k} \cdot \mathbf{r}}\,d\mathbf{k} \nonumber \\
&= \int_0^\infty dk\,k^2 \left(\frac{2 E(k)}{4\pi k^2}\right) \int_0^{2\pi} d\phi \int_0^\pi e^{i kr \cos\theta} \sin\theta\,d\theta \nonumber \\
&= \int_0^\infty dk\,k^2 \left(\frac{E(k)}{2\pi k^2}\right) (2\pi) \left[ \frac{e^{i kr} - e^{-i kr}}{i kr} \right] \nonumber \\
&= 2 \int_0^\infty E(k)\,\frac{\sin(kr)}{kr}\,dk.
\label{eq:trace_rep}
\end{align}
Eq.~\eqref{eq:trace_rep} exactly matches Schoenberg's representation~\eqref{eq:schoenberg_rep} in $\mathbb{R}^3$ with the non-negative radial measure $d\nu(k) = 2E(k)\,dk \ge 0$.

Similarly, for the longitudinal correlation component $R_{11}(r \mathbf{e}_1) = u'^2 f(r)$, we evaluate the inverse Fourier transform of $\Phi_{11}(\mathbf{k}) = \frac{E(k)}{4\pi k^4}(k^2 - k_1^2)$ over spherical coordinates with polar angle $\theta$ measured from the separation axis $\mathbf{e}_1$ ($k_1 = k\cos\theta$, $\mathbf{k}\cdot\mathbf{r} = kr\cos\theta$):
\begin{align}
u'^2 f(r) &= \iiint_{\mathbb{R}^3} \Phi_{11}(\mathbf{k})\,e^{i k_1 r}\,d\mathbf{k} \nonumber \\
&= \int_0^\infty dk\,k^2 \frac{E(k)}{4\pi k^4} \int_0^{2\pi} d\phi \int_0^\pi k^2 (1 - \cos^2\theta)\,e^{i kr \cos\theta}\,\sin\theta\,d\theta \nonumber \\
&= \int_0^\infty dk\,E(k) \frac{1}{2} \int_{-1}^1 (1 - \mu^2)\,e^{i (kr) \mu}\,d\mu,
\end{align}
where $\mu = \cos\theta$. Integrating by parts with respect to $\mu$:
\begin{equation}
\frac{1}{2} \int_{-1}^1 (1 - \mu^2) \cos(x\mu)\,d\mu = 2 \left( \frac{\sin x - x \cos x}{x^3} \right).
\end{equation}
Therefore, we obtain the exact spherical projection representation:
\begin{equation}
u'^2 f(r) = \int_0^\infty E(k)\,K_{11}(kr)\,dk,
\label{eq:f_kernel_rep}
\end{equation}
where the longitudinal projection kernel $K_{11}(x)$ is
\begin{equation}
K_{11}(x) = 2 \left( \frac{\sin x - x\cos x}{x^3} \right) = 2\,\frac{j_1(x)}{x},
\label{eq:K11_kernel}
\end{equation}
with $j_1(x) = (\sin x - x\cos x)/x^2$ being the spherical Bessel function of order 1.

\begin{figure}[t]
\centering
\includegraphics[width=0.72\textwidth]{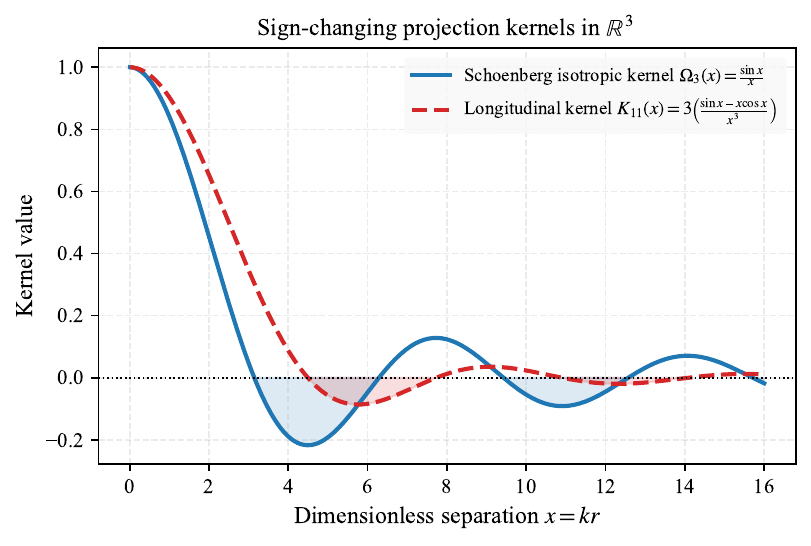}
\caption{The isotropic Schoenberg projection kernel $\Omega_3(x) = \sin(x)/x$ and the longitudinal projection kernel $K_{11}(x) = 2(\sin x - x\cos x)/x^3$ on $\mathbb{R}^3$. Both kernels oscillate and exhibit pronounced negative excursions, proving that positive-definiteness on $\mathbb{R}^3$ does not imply pointwise non-negativity of the physical-space correlation functions.}
\label{fig:schoenberg_kernel}
\end{figure}

Both kernels $\Omega_3(x)$ and $K_{11}(x)$ are plotted in Figure~\ref{fig:schoenberg_kernel}. As is evident:
\begin{itemize}
\item $\Omega_3(x)$ possesses its first zero at $x = \pi \approx 3.1416$ and attains a minimum of $\approx -0.2172$ at $x \approx 4.4934$.
\item $K_{11}(x)$ crosses zero at $x \approx 4.4934$ and attains a negative minimum of $\approx -0.0543$ near $x \approx 5.7635$.
\end{itemize}

This comparison clarifies the critical distinction between positive-definiteness and pointwise non-negativity:
\begin{enumerate}
\item \textbf{Positive-definiteness is guaranteed:} By Bochner's and Schoenberg's theorems, $E(k) \ge 0$ rigorously guarantees that $R_{ij}(\mathbf{r})$ is a positive-definite tensor field. This enforces standard algebraic Gram-matrix inequalities such as boundedness ($|f(r)| \le f(0) = 1$) and origin concavity ($f''(0) \le 0$).
\item \textbf{Pointwise non-negativity is not guaranteed:} Because the projection kernels $\Omega_3(x)$ and $K_{11}(x)$ are sign-changing and exhibit substantial negative excursions, integrating them against a strictly non-negative measure $E(k) \ge 0$ does \emph{not} force either the trace $f(r) + 2g(r)$ or the longitudinal correlation $f(r)$ to remain non-negative pointwise.
\item \textbf{Transverse vs. Longitudinal behavior:} For the transverse correlation $g(r)$, incompressibility kinematically demands a negative loop ($\int_0^\infty r g(r)\,dr = 0$, Eq.~\eqref{eq:g_zero_integral}). For the longitudinal correlation $f(r)$, three-dimensional positive-definiteness structurally fails to preclude negative values.
\end{enumerate}
Consequently, multidimensional harmonic analysis alone cannot resolve whether $f(r) \ge 0$, and an alternative one-dimensional formulation is required.

\section{The One-Dimensional Spectrum and P\'olya's Criterion}
\label{sec:polya_framework}

In this section, we reduce the correlation problem to a one-dimensional Fourier cosine transform, prove the unconditional monotonicity of the longitudinal spectrum $E_{11}(k_1)$, and establish P\'olya's convexity criterion as a rigorous sufficient condition for the non-negativity of $f(r)$.

\subsection{The 1D Longitudinal Spectrum \texorpdfstring{$E_{11}(k_1)$}{E11(k1)}}
\label{subsec:1d_spectrum}

The resolution of the sign problem lies in reducing the three-dimensional geometry to an exact one-dimensional Fourier transform. Rather than performing a spherical angular reduction (which introduces the oscillatory Bessel-type kernel $K_{11}$ in Eq.~\eqref{eq:f_kernel_rep}), we can obtain the 1D relation directly from the three-dimensional Fourier transform through tensor contraction and transverse projection.

Starting from the general 3D inverse Fourier representation of the correlation tensor:
\begin{equation}
R_{ij}(\mathbf{r}) = \iiint_{\mathbb{R}^3} \Phi_{ij}(\mathbf{k})\,e^{i \mathbf{k} \cdot \mathbf{r}}\,d\mathbf{k},
\end{equation}
we contract the tensor along the longitudinal unit vector $\mathbf{e}_1$ (taking the $(1,1)$ component) and evaluate it along the separation axis $\mathbf{r} = r \mathbf{e}_1$:
\begin{equation}
u'^2 f(r) = R_{11}(r\mathbf{e}_1) = \iiint_{\mathbb{R}^3} \Phi_{11}(k_1, k_2, k_3)\,e^{i k_1 r}\,dk_1\,dk_2\,dk_3.
\end{equation}
Because the phase factor $e^{i k_1 r}$ depends solely on the longitudinal wavenumber $k_1$, we can separate the longitudinal and transversal integrations:
\begin{equation}
u'^2 f(r) = \int_{-\infty}^\infty e^{i k_1 r} \left[ \iint_{\mathbb{R}^2} \Phi_{11}(k_1, k_2, k_3)\,dk_2\,dk_3 \right] dk_1.
\label{eq:f_proj_integral}
\end{equation}
The term inside the brackets is the planar projection of $\Phi_{11}(\mathbf{k})$ onto the $k_1$-axis across the transversal wavenumber plane $\mathbf{k}_\perp = (k_2, k_3)$. Defining the one-sided, one-dimensional longitudinal energy spectrum $E_{11}(k_1)$ for $k_1 \ge 0$ as
\begin{equation}
E_{11}(k_1) \equiv 2 \iint_{\mathbb{R}^2} \Phi_{11}(k_1, k_2, k_3)\,dk_2\,dk_3,
\label{eq:E11_def_plane_integral}
\end{equation}
and noting that $\Phi_{11}(\mathbf{k})$ is even in $k_1$, Eq.~\eqref{eq:f_proj_integral} reduces exactly to:
\begin{equation}
u'^2 f(r) = \int_0^\infty E_{11}(k_1)\,\cos(k_1 r)\,dk_1.
\label{eq:f_cos_E11}
\end{equation}

Equivalently, one can arrive at this identical, cosine transform pair by applying the classical 1D Wiener-Khinchin theorem directly to the one-dimensional stationary longitudinal velocity signal $u_1(x_1)$ recorded along a spatial streamline (e.g., as measured in experimental fluid mechanics via hot-wire anemometry and Taylor's frozen-flow hypothesis \cite{Taylor1935, Pope2000}), which immediately gives the inverse transform:
\begin{equation}
E_{11}(k_1) = \frac{2 u'^2}{\pi} \int_0^\infty f(r)\,\cos(k_1 r)\,dr.
\label{eq:E11_cos_f}
\end{equation}
Both formulations are kinematically exact and equivalent representations of the longitudinal correlation.

To express $E_{11}(k_1)$ in terms of the three-dimensional isotropic energy spectrum $E(k)$, we substitute the isotropic form $\Phi_{11}(\mathbf{k}) = \frac{E(k)}{4\pi k^4}(k^2 - k_1^2)$ from Eq.~\eqref{eq:Phi_ij_isotropic} into Eq.~\eqref{eq:E11_def_plane_integral}:
\begin{equation}
E_{11}(k_1) = 2 \iint_{\mathbb{R}^2} \frac{E(k)}{4\pi k^4} (k^2 - k_1^2)\,dk_2\,dk_3.
\end{equation}
Transforming the $(k_2, k_3)$ plane to polar coordinates $(k_\perp, \phi)$ where $k_\perp^2 = k_2^2 + k_3^2 = k^2 - k_1^2$ and $dk_2\,dk_3 = k_\perp\,dk_\perp\,d\phi = \frac{1}{2}\,d(k_\perp^2)\,d\phi = k\,dk\,d\phi$ (with $k$ ranging from $k_1$ to $\infty$):
\begin{equation}
E_{11}(k_1) = 2 \int_0^{2\pi} d\phi \int_{k_1}^\infty \frac{E(k)}{4\pi k^4} (k^2 - k_1^2)\,k\,dk = \int_{k_1}^\infty \frac{E(k)}{k} \left( 1 - \frac{k_1^2}{k^2} \right)\,dk.
\label{eq:E11_from_Ek}
\end{equation}

\subsection{Unconditional Monotonicity of \texorpdfstring{$E_{11}(k_1)$}{E11(k1)}}

We now establish a general mathematical property of the 1D longitudinal spectrum.

\begin{theorem}[Unconditional Monotonicity]
\label{thm:monotonicity}
For any kinematically admissible, non-negative three-dimensional energy spectrum $E(k) \ge 0$ with $\int_0^\infty E(k)\,dk < \infty$, the one-dimensional longitudinal spectrum $E_{11}(k_1)$ is strictly monotonically decreasing for all $k_1 > 0$:
\begin{equation}
\frac{d E_{11}(k_1)}{d k_1} \le 0.
\label{eq:thm_mono}
\end{equation}
\end{theorem}

\begin{proof}
Differentiating Eq.~\eqref{eq:E11_from_Ek} with respect to $k_1$ using Leibniz's rule:
\begin{equation}
\frac{d E_{11}(k_1)}{d k_1} = - \left[ \frac{E(k)}{k} \left( 1 - \frac{k_1^2}{k^2} \right) \right]_{k = k_1} + \int_{k_1}^\infty \frac{\partial}{\partial k_1} \left[ \frac{E(k)}{k} \left( 1 - \frac{k_1^2}{k^2} \right) \right] dk.
\end{equation}
The boundary term at $k = k_1$ evaluates to $(E(k_1)/k_1)(1 - 1) = 0$. The derivative under the integral is
\begin{equation}
\frac{\partial}{\partial k_1} \left( 1 - \frac{k_1^2}{k^2} \right) = -\frac{2 k_1}{k^2}.
\end{equation}
Thus, we obtain the exact closed-form first derivative:
\begin{equation}
\frac{d E_{11}(k_1)}{d k_1} = -2 k_1 \int_{k_1}^\infty \frac{E(k)}{k^3}\,dk.
\label{eq:dE11_exact}
\end{equation}
Since $E(k) \ge 0$ by kinematic realizability and $k_1 > 0$, the integrand $E(k)/k^3 \ge 0$ everywhere on $[k_1, \infty)$. Therefore, the integral is non-negative, and the factor $-2k_1 < 0$ strictly ensures
\begin{equation}
\frac{d E_{11}(k_1)}{d k_1} \le 0 \quad \text{for all } k_1 \ge 0.
\end{equation}
Monotonicity is unconditional; it holds for every realizable turbulent state without exception.
\end{proof}

\subsection{Convexity and P\'olya's Theorem}

We now examine the second derivative of $E_{11}(k_1)$.

\begin{theorem}[Second Derivative of $E_{11}$]
\label{thm:second_derivative}
The second derivative of the 1D longitudinal spectrum is given by
\begin{equation}
\frac{d^2 E_{11}(k_1)}{d k_1^2} = \frac{2 E(k_1)}{k_1^2} - 2 \int_{k_1}^\infty \frac{E(k)}{k^3}\,dk.
\label{eq:d2E11_exact}
\end{equation}
\end{theorem}

\begin{proof}
Differentiating Eq.~\eqref{eq:dE11_exact} with respect to $k_1$ using the product rule and Leibniz's rule:
\begin{align}
\frac{d^2 E_{11}(k_1)}{d k_1^2} &= \frac{d}{dk_1} \left[ -2 k_1 \int_{k_1}^\infty \frac{E(k)}{k^3}\,dk \right] \nonumber \\
&= -2 \int_{k_1}^\infty \frac{E(k)}{k^3}\,dk - 2 k_1 \left( - \frac{E(k_1)}{k_1^3} \right) \nonumber \\
&= \frac{2 E(k_1)}{k_1^2} - 2 \int_{k_1}^\infty \frac{E(k)}{k^3}\,dk.
\end{align}
\end{proof}

Unlike the first derivative, the two terms in Eq.~\eqref{eq:d2E11_exact} have \emph{opposite signs}: the local term $2E(k_1)/k_1^2$ is positive, while the non-local tail integral $-2\int_{k_1}^\infty E(k)/k^3\,dk$ is negative. Thus, convexity ($d^2 E_{11}/dk_1^2 \ge 0$) is \textbf{not automatic}; it represents a non-trivial shape condition on $E(k)$.

The significance of convexity is established by P\'olya's classical theorem in harmonic analysis \cite{Polya1949}:

\begin{theorem}[P\'olya's Criterion (1918, 1949)]
\label{thm:polya}
Let $h(k)$ be a real-valued, even, continuous function on $(-\infty, \infty)$ such that:
\begin{enumerate}
\item $h(k) \ge 0$ for all $k \ge 0$,
\item $h(k)$ is monotonically decreasing on $[0, \infty)$,
\item $h(k)$ is convex on $[0, \infty)$ (i.e., $h''(k) \ge 0$ in the sense of distributions),
\item $\lim_{k \to \infty} h(k) = 0$.
\end{enumerate}
Then its Fourier transform (cosine transform),
\begin{equation}
\hat{h}(r) = \int_0^\infty h(k)\,\cos(kr)\,dk,
\end{equation}
is strictly non-negative everywhere: $\hat{h}(r) \ge 0$ for all $r \in [0, \infty)$.
\end{theorem}

Combining Theorem~\ref{thm:monotonicity}, Theorem~\ref{thm:second_derivative}, and Theorem~\ref{thm:polya} yields an immediate, rigorous sufficient condition:

\begin{corollary}[Sufficient Condition for Non-Negativity]
\label{cor:sufficient}
If the 3D energy spectrum $E(k)$ satisfies the convexity inequality
\begin{equation}
\frac{E(k_1)}{k_1^2} \ge \int_{k_1}^\infty \frac{E(k)}{k^3}\,dk \quad \text{for all } k_1 > 0,
\label{eq:convexity_condition}
\end{equation}
then the longitudinal autocorrelation function $f(r)$ is strictly non-negative for all $r \ge 0$.
\end{corollary}

\subsection{Convexity of Power-Law Spectra and the Kolmogorov Inertial Range}

We can utilize Eq.~\eqref{eq:convexity_condition} on a general power-law spectrum parameterized by the spectral decay exponent $s$:
\begin{equation}
E(k) = C\,k^{-s} \quad (C > 0).
\label{eq:general_power_law}
\end{equation}
For the upper-tail integral in Eq.~\eqref{eq:convexity_condition} to converge as $k \to \infty$, the integrand must decay faster than $k^{-1}$, which imposes the convergence requirement $s+3 > 1$, or equivalently:
\begin{equation}
s > -2 \quad (s+2 > 0).
\label{eq:convergence_req}
\end{equation}
Under this condition, evaluating the two terms in Eq.~\eqref{eq:convexity_condition} gives:
\begin{align}
\frac{E(k_1)}{k_1^2} &= C\,k_1^{-(s+2)}, \\
\int_{k_1}^\infty \frac{E(k)}{k^3}\,dk &= C \int_{k_1}^\infty k^{-(s+3)}\,dk = \frac{C}{s+2}\,k_1^{-(s+2)}.
\end{align}
Substituting these expressions into the second derivative Eq.~\eqref{eq:d2E11_exact} yields
\begin{equation}
\frac{d^2 E_{11}(k_1)}{d k_1^2} = 2 C\,k_1^{-(s+2)} \left( 1 - \frac{1}{s+2} \right) = 2 C\,k_1^{-(s+2)} \left( \frac{s+1}{s+2} \right).
\label{eq:power_law_convex}
\end{equation}
Because convergence at infinity guarantees $s+2 > 0$, the sign of the second derivative is determined entirely by the numerator:
\begin{equation}
\frac{d^2 E_{11}(k_1)}{d k_1^2} \ge 0 \quad \Longleftrightarrow \quad s \ge -1 \quad (\text{or } p \le 1 \text{ for } E(k) \propto k^p).
\label{eq:exponent_bound}
\end{equation}
(Note that although the algebraic ratio $(s+1)/(s+2) \ge 0$ also formally holds for $s < -2$, that branch corresponds to ultraviolet divergence of the integral where the power-law cannot extend to $k \to \infty$.)

When combined with the requirements of finite total energy and vanishing spectral tail ($\lim_{k \to \infty} E_{11}(k) = 0$, which requires $s > 0$), any physically decaying cascade scaling ($s > 0$) strictly and unconditionally satisfies $s \ge -1$.

In particular, for Kolmogorov's 1941 inertial-range scaling:
\begin{equation}
E(k) = C_K\,\epsilon^{2/3}\,k^{-5/3} \quad (s = 5/3),
\label{eq:kolmogorov_spectrum}
\end{equation}
where $C_K \approx 1.5$ is the Kolmogorov constant and $\epsilon$ is the kinetic energy dissipation rate per unit mass, the second derivative evaluates to
\begin{equation}
\frac{d^2 E_{11}(k_1)}{d k_1^2} = 2 C_K\,\epsilon^{2/3}\,k_1^{-11/3} \left( 1 - \frac{3}{11} \right) = \frac{16}{11}\,C_K\,\epsilon^{2/3}\,k_1^{-11/3} > 0.
\label{eq:kolmogorov_convex}
\end{equation}
Because $s = 5/3 > -1$ (i.e., $1 \ge 3/11$), the Kolmogorov spectrum is \textbf{strictly and robustly convex} throughout the inertial range. If a turbulent spectrum consisted purely of decaying cascade scaling with $s \ge -1$, P\'olya's criterion would guarantee $f(r) \ge 0$ unconditionally.

\section{Infrared Constraints and Near-Origin Concavity}
\label{sec:infrared}

In this section, we examine large-scale hydrodynamical invariants and prove that physically realizable turbulent spectra belonging to both the Batchelor--Proudman ($E(k) \propto k^4$) and Saffman ($E(k) \propto k^2$) infrared universality classes, are necessarily non-convex (strictly concave) in a finite neighborhood of the origin $k_1 = 0$. We then show when and how the convexity criterion is obeyed in the inertial range, and how broadband inertial scaling overcomes this infrared concavity in mature turbulence.

\subsection{Large-Scale Invariants and Infrared Expansions}

In any real physical fluid, the Kolmogorov cascade scaling cannot extend down to $k \to 0$, as that would imply divergent total turbulent kinetic energy. The asymptotic behavior of $E(k)$ as $k \to 0$ is dictated by the conservation of large-scale hydrodynamical invariants under the Navier--Stokes equations \cite{BatchelorProudman1956, Saffman1967, Davidson2004}:
\begin{enumerate}
\item \textbf{Saffman Turbulence ($E(k) \sim k^2$):} When the initial flow possesses a non-zero linear momentum or non-vanishing long-range velocity impulse correlations, the Saffman invariant $L_S = \frac{1}{4\pi} \int_{\mathbb{R}^3} \langle \mathbf{u}(\mathbf{x}) \cdot \mathbf{u}(\mathbf{x}+\mathbf{r}) \rangle\,d\mathbf{r}$ is an exact invariant of the decaying Navier--Stokes equations \cite{Saffman1967, Davidson2004}. Taylor-series expansion of the spectral tensor near $\mathbf{k} = \mathbf{0}$ yields the infrared asymptotic expansion:
\begin{equation}
E(k) = C_S\,k^2 + C_{S,4}\,k^4 + O(k^6) \quad (k \to 0), \qquad C_S = 4\pi L_S > 0.
\label{eq:saffman_expansion}
\end{equation}
\item \textbf{Batchelor--Proudman Turbulence ($E(k) \sim k^4$):} When turbulence is generated by localized, momentum-free stirring (such as passive grids in wind tunnels or impulsively stirred finite domains), the linear impulse identically vanishes ($L_S = 0$). Incompressible non-local pressure fluctuations generate quadrupolar long-range velocity correlations ($\sim r^{-5}$), establishing the invariance of Loitsianskii's integral $\Lambda = u'^2 \int_0^\infty r^4 f(r)\,dr$ and fixing the infrared expansion \cite{BatchelorProudman1956, Davidson2004}:
\begin{equation}
E(k) = C_B\,k^4 + C_{B,6}\,k^6 + O(k^8) \quad (k \to 0), \qquad C_B = \frac{2}{3\pi} \Lambda > 0.
\label{eq:batchelor_expansion}
\end{equation}
\end{enumerate}

\subsection{Evaluation of the Convexity Criterion for Physical Infrared Spectra}

Recall from Corollary~\ref{cor:sufficient} and Eq.~\eqref{eq:d2E11_exact} that the exact second derivative governing the convexity of $E_{11}(k_1)$ is
\begin{equation}
\frac{d^2 E_{11}(k_1)}{d k_1^2} = 2 \left[ \frac{E(k_1)}{k_1^2} - \int_{k_1}^\infty \frac{E(k)}{k^3}\,dk \right].
\label{eq:convexity_diff}
\end{equation}
We now prove that \emph{neither} physical infrared family satisfies P\'olya's convexity criterion globally on $[0, \infty)$, because both families strictly violate convexity in an infrared neighborhood of the origin.

\begin{theorem}[Near-Origin Concavity of Physical Turbulence Spectra]
\label{thm:infrared_concavity}
Let $E(k) \ge 0$ be physically admissible energy spectrum of homogeneous isotropic turbulence with finite total energy $\mathcal{K} = \int_0^\infty E(k)\,dk < \infty$ of the following categories:
\begin{enumerate}
\item \textbf{Batchelor Family ($p = 4$):} For any spectrum satisfying Batchelor's infrared scaling $E(k) = C_B k^4 + O(k^6)$ as $k \to 0$, $E_{11}(k_1)$ is strictly concave on a finite interval $(0, k_c)$ with a finite negative origin limit:
\begin{equation}
\lim_{k_1 \to 0^+} \frac{d^2 E_{11}(k_1)}{d k_1^2} = -2 \int_0^\infty \frac{E(k)}{k^3}\,dk = -2 I_0 < 0.
\label{eq:batchelor_limit}
\end{equation}
Near the origin, the exact asymptotic behavior is:
\begin{equation}
\frac{d^2 E_{11}(k_1)}{d k_1^2} = -2 I_0 + 3 C_B\,k_1^2 + O(k_1^4) < 0 \quad \text{for all } k_1 \in \big(0, \sqrt{2I_0/(3C_B)}\,\big).
\label{eq:batchelor_asymptotic}
\end{equation}

\item \textbf{Saffman Family ($p = 2$):} For any spectrum satisfying Saffman's infrared scaling $E(k) = C_S k^2 + C_{S,4} k^4 + O(k^6)$ as $k \to 0$ with $C_S > 0$, $E_{11}(k_1)$ is strictly concave on a finite interval $(0, k_c)$ with a logarithmic infrared divergence to $-\infty$:
\begin{equation}
\frac{d^2 E_{11}(k_1)}{d k_1^2} = 2 C_S\,\ln k_1 + 2(C_S - J_0) + O(k_1^2) \longrightarrow -\infty \quad \text{as } k_1 \to 0^+,
\label{eq:saffman_limit}
\end{equation}
where $J_0 = \int_0^{k_0} \frac{E(k) - C_S k^2}{k^3}\,dk + \int_{k_0}^\infty \frac{E(k)}{k^3}\,dk - C_S \ln k_0$ is a finite constant independent of $k_1$.
\end{enumerate}
Consequently, spectra of both families are strictly concave in an infrared neighborhood $(0, k_c)$ of the origin, and global convexity on $[0, \infty)$ is unconditionally violated.
\end{theorem}

\begin{proof}
We evaluate the two terms in Eq.~\eqref{eq:convexity_diff} separately for each infrared family:

\medskip
\noindent\textbf{Case 1: Batchelor Family ($E(k) = C_B k^4 + O(k^6)$).}
The first term evaluates to
\begin{equation}
\frac{E(k_1)}{k_1^2} = C_B\,k_1^2 + O(k_1^4) \implies \lim_{k_1 \to 0^+} \frac{E(k_1)}{k_1^2} = 0.
\end{equation}
For the upper-tail integral, the integral converges to a strictly positive, finite constant $I_0 = \int_0^\infty \frac{E(k)}{k^3}\,dk > 0$. The finiteness of $I_0$ is guaranteed at both integration limits:
\begin{itemize}
\item \emph{Infrared limit ($k \to 0$):} Since $E(k)/k^3 = C_B k + O(k^3)$, the integrand is continuous and vanishes linearly at $k=0$, ensuring absolute integrability on $[0, k_0]$.
\item \emph{Ultraviolet limit ($k \to \infty$):} Because every physically realizable turbulent flow has finite total turbulent kinetic energy per unit mass ($\mathcal{K} = \int_0^\infty E(k)\,dk < \infty$), the ultraviolet tail is strictly bounded for any $k_0 > 0$ via the comparison test:
\begin{equation}
\int_{k_0}^\infty \frac{E(k)}{k^3}\,dk \le \frac{1}{k_0^3} \int_{k_0}^\infty E(k)\,dk = \frac{\mathcal{K}_{\text{UV}}}{k_0^3} < \infty.
\end{equation}
\end{itemize}
Splitting the integral at $k_1$:
\begin{equation}
\int_{k_1}^\infty \frac{E(k)}{k^3}\,dk = I_0 - \int_0^{k_1} \left( C_B k + O(k^3) \right) dk = I_0 - \frac{1}{2} C_B k_1^2 + O(k_1^4).
\end{equation}
Substituting into Eq.~\eqref{eq:convexity_diff} yields:
\begin{equation}
\frac{d^2 E_{11}(k_1)}{d k_1^2} = 2 \left( C_B k_1^2 - I_0 + \frac{1}{2} C_B k_1^2 \right) + O(k_1^4) = -2 I_0 + 3 C_B k_1^2 + O(k_1^4).
\end{equation}
Taking $k_1 \to 0^+$ gives $\lim_{k_1 \to 0^+} d^2 E_{11}/dk_1^2 = -2 I_0 < 0$. By continuity, $d^2 E_{11}/dk_1^2 < 0$ on $(0, k_c)$ with $k_c \approx \sqrt{2I_0/(3C_B)}$.

\medskip
\noindent\textbf{Case 2: Saffman Family ($E(k) = C_S k^2 + C_{S,4} k^4 + O(k^6)$ with $C_S > 0$).}

We evaluate each term of Eq.~\eqref{eq:convexity_diff} step by step:

\paragraph{Step 1: Evaluate the local term $E(k_1)/k_1^2$.}
Dividing the Saffman Taylor expansion by $k_1^2$ gives
\begin{equation}
\frac{E(k_1)}{k_1^2} = \frac{C_S k_1^2 + C_{S,4} k_1^4 + O(k_1^6)}{k_1^2} = C_S + C_{S,4}\,k_1^2 + O(k_1^4).
\label{eq:saffman_term1}
\end{equation}
As $k_1 \to 0^+$, this term approaches the finite positive constant $C_S > 0$.

\paragraph{Step 2: Decompose the tail integral $\int_{k_1}^\infty [E(k)/k^3]\,dk$.}
Near $k = 0$, the integrand behaves as $E(k)/k^3 \approx C_S / k$, which produces a logarithmic singularity at $k = 0$. To isolate this singularity rigorously, choose a fixed intermediate wavenumber $k_0 > 0$ with $k_1 < k_0$ and split the domain of integration into two regions:
\begin{equation}
\int_{k_1}^\infty \frac{E(k)}{k^3}\,dk = \int_{k_1}^{k_0} \frac{E(k)}{k^3}\,dk + \int_{k_0}^\infty \frac{E(k)}{k^3}\,dk.
\label{eq:saffman_split}
\end{equation}

\paragraph{Step 3: Extract the logarithmic divergence.}
Inside the low-wavenumber interval $[k_1, k_0]$, add and subtract the leading $C_S k^2$ term:
\begin{equation}
\int_{k_1}^{k_0} \frac{E(k)}{k^3}\,dk = \int_{k_1}^{k_0} \frac{C_S k^2}{k^3}\,dk + \int_{k_1}^{k_0} \frac{E(k) - C_S k^2}{k^3}\,dk.
\end{equation}
Evaluating the first integral analytically:
\begin{equation}
\int_{k_1}^{k_0} \frac{C_S k^2}{k^3}\,dk = C_S \int_{k_1}^{k_0} \frac{dk}{k} = C_S \Big[ \ln k_0 - \ln k_1 \Big] = -C_S \ln k_1 + C_S \ln k_0.
\end{equation}
For the second integral, the subtracted integrand $(E(k) - C_S k^2)/k^3 = C_{S,4} k + O(k^3)$ is continuous and regular at $k = 0$. Therefore:
\begin{align}
\int_{k_1}^{k_0} \frac{E(k) - C_S k^2}{k^3}\,dk &= \int_0^{k_0} \frac{E(k) - C_S k^2}{k^3}\,dk - \int_0^{k_1} \Big( C_{S,4} k + O(k^3) \Big) dk \nonumber \\
&= \int_0^{k_0} \frac{E(k) - C_S k^2}{k^3}\,dk - \frac{1}{2} C_{S,4} k_1^2 + O(k_1^4).
\end{align}

\paragraph{Step 4: Combine into the finite constant $J_0$.}
Define the $k_1$-independent constant $J_0$:
\begin{equation}
J_0 \equiv C_S \ln k_0 + \int_0^{k_0} \frac{E(k) - C_S k^2}{k^3}\,dk + \int_{k_0}^\infty \frac{E(k)}{k^3}\,dk.
\label{eq:J0_def}
\end{equation}
Note that $J_0$ is strictly finite (since the first integral is regularized and the upper tail is bounded by $k_0^{-3}\mathcal{K} < \infty$) and is independent of the choice of $k_0$ (differentiating $J_0$ with respect to $k_0$ gives zero identically). Thus:
\begin{equation}
\int_{k_1}^\infty \frac{E(k)}{k^3}\,dk = -C_S \ln k_1 + J_0 - \frac{1}{2} C_{S,4} k_1^2 + O(k_1^4).
\label{eq:saffman_term2}
\end{equation}

\paragraph{Step 5: Compute the net second derivative $d^2 E_{11}/dk_1^2$.}
Subtracting Eq.~\eqref{eq:saffman_term2} from Eq.~\eqref{eq:saffman_term1} and multiplying by 2 yields:
\begin{align}
\frac{d^2 E_{11}(k_1)}{d k_1^2} &= 2 \left[ \left( C_S + C_{S,4} k_1^2 \right) - \left( -C_S \ln k_1 + J_0 - \frac{1}{2} C_{S,4} k_1^2 \right) \right] + O(k_1^4) \nonumber \\
&= 2 C_S \ln k_1 + 2(C_S - J_0) + 3 C_{S,4} k_1^2 + O(k_1^4).
\label{eq:saffman_d2E11_exact}
\end{align}
Because $C_S > 0$ and $\lim_{k_1 \to 0^+} \ln k_1 = -\infty$, taking the limit as $k_1 \to 0^+$ gives:
\begin{equation}
\lim_{k_1 \to 0^+} \frac{d^2 E_{11}(k_1)}{d k_1^2} = -\infty.
\end{equation}
Thus, $E_{11}(k_1)$ is \textbf{unconditionally concave} in an infrared interval $(0, k_c)$ near the origin, with crossover wavenumber $k_c \approx \exp\big((J_0 - C_S)/C_S\big)$.
\end{proof}

\subsection{When is the Convexity Criterion Obeyed and How Positivity is Maintained?}

Theorem~\ref{thm:infrared_concavity} establishes that \textbf{neither the Batchelor nor the Saffman family satisfies P\'olya's convexity criterion in the infrared range ($k_1 < k_c$)}. We now clarify precisely where the convexity criterion holds and how the net sign of the cosine transform is determined:

\begin{enumerate}
\item \textbf{Inertial Cascade Range ($k_1 \gg k_c$):} When the spectrum enters the decaying cascade scaling $E(k) \propto k^{-s}$ with $s > 0$ (such as Kolmogorov $s = 5/3$), Eq.~\eqref{eq:power_law_convex} proves that
\begin{equation}
\frac{d^2 E_{11}(k_1)}{d k_1^2} = 2 C\,k_1^{-(s+2)} \left( \frac{s+1}{s+2} \right) > 0 \quad \text{for all } s > -1.
\end{equation}
Thus, throughout the entire inertial and dissipation ranges, $E_{11}(k_1)$ is \textbf{strictly convex}.

\item \textbf{Crossover Wavenumber $k_c$:} The transition from infrared concavity ($d^2 E_{11}/dk_1^2 < 0$) to inertial convexity ($d^2 E_{11}/dk_1^2 > 0$) occurs at an intermediate inflection point $k_1 = k_c$, where $E(k_c)/k_c^2 = \int_{k_c}^\infty [E(k)/k^3]\,dk$.

\item \textbf{Net Sign via Second-Derivative Integral Representation:} Integrating the cosine transform $\int_0^\infty E_{11}(k_1)\cos(k_1 r)\,dk_1$ by parts twice (noting that $E_{11}'(0) = 0$ and $E_{11}(\infty) = E_{11}'(\infty) = 0$) yields the exact identity from \eqref{eq:f_cos_E11}:
\begin{equation}
u'^2 f(r) = \frac{1}{r^2} \int_0^\infty \frac{d^2 E_{11}(k_1)}{d k_1^2} \big(1 - \cos(k_1 r)\big)\,dk_1.
\label{eq:f_by_parts_second_deriv}
\end{equation}
Because the kernel $1 - \cos(k_1 r) = 2\sin^2(k_1 r / 2) \ge 0$ is strictly non-negative, the net sign of $f(r)$ represents a direct competition between the negative contribution from the concave infrared region $[0, k_c]$ and the positive contribution from the convex cascade region $[k_c, \infty)$:
\begin{equation}
u'^2 f(r) = \frac{1}{r^2} \left[ \underbrace{\int_0^{k_c} \frac{d^2 E_{11}}{dk_1^2} \big(1 - \cos(k_1 r)\big)\,dk_1}_{< 0 \text{ (concave infrared)}} + \underbrace{\int_{k_c}^\infty \frac{d^2 E_{11}}{dk_1^2} \big(1 - \cos(k_1 r)\big)\,dk_1}_{> 0 \text{ (convex cascade)}} \right].
\end{equation}
Because P\'olya's condition is sufficient but not necessary, failure of global convexity does not imply that $f(r)$ must become negative. In mature broadband turbulence, the positive integral over the broad Kolmogorov cascade overwhelmingly outweighs the sub-integral infrared deficit, ensuring that the net transform $f(r)$ remains strictly positive for all $r$, as confirmed by direct high-precision quadrature in the next section.
\end{enumerate}

\begin{figure}[t]
\centering
\includegraphics[width=0.88\textwidth]{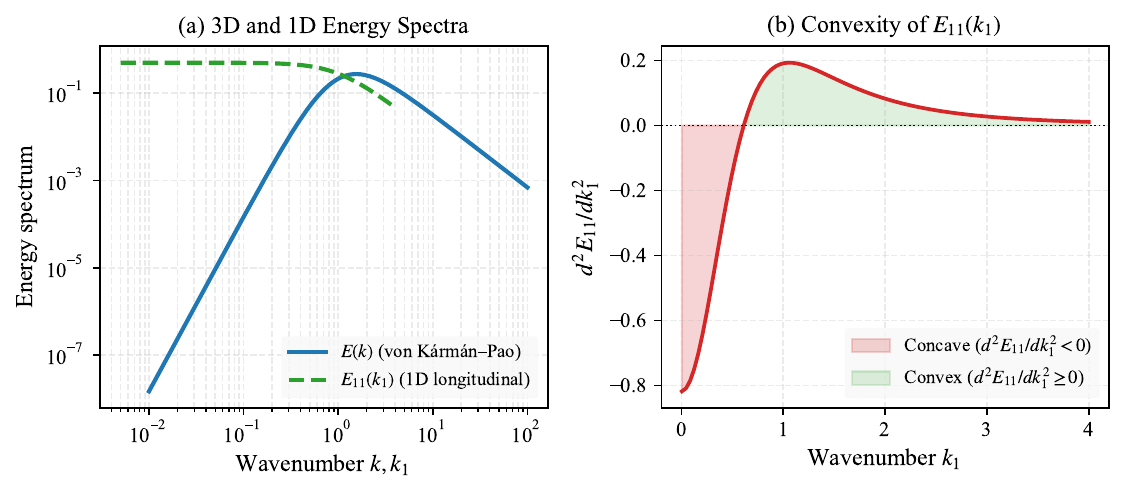}
\caption{(a) The 3D energy spectrum $E(k)$ for the von K\'arm\'an--Pao model with $E(k) \sim k^4$ as $k \to 0$ and $E(k) \sim k^{-5/3}$ in the inertial range, alongside the resulting 1D longitudinal spectrum $E_{11}(k_1)$. (b) The second derivative $d^2E_{11}/dk_1^2$, illustrating the finite concave region ($k_1 \lesssim 0.57$) near the origin resulting from Batchelor infrared rolloff, and the convex region ($k_1 \gtrsim 0.57$) dominated by the inertial cascade.}
\label{fig:vkp_convexity}
\end{figure}

\subsection{Analysis of the von K\'arm\'an--Pao Model Spectrum}

To assess the quantitative impact of this infrared concavity in a realistic broadband turbulent field, we analyze the standard von K\'arm\'an--Pao model spectrum \cite{vonKarman1948, Pao1965, Pope2000}:
\begin{equation}
E(k) = \alpha\,u'^2\,L\,\frac{(kL)^4}{\big(1 + (kL)^2\big)^{17/6}}\,\exp\left(-\frac{3}{2} C_K (k\eta)^{4/3}\right),
\label{eq:vkp_model}
\end{equation}
where $L$ is the integral length scale, $\eta = (\nu^3/\epsilon)^{1/4}$ is the Kolmogorov dissipation scale, and $\alpha \approx 1.5$. This spectrum incorporates:
\begin{itemize}
\item Batchelor $k^4$ infrared scaling for $kL \ll 1$,
\item Kolmogorov $k^{-5/3}$ inertial-range power law for $1/L \ll k \ll 1/\eta$,
\item Pao exponential viscous cutoff for $k\eta \gtrsim 1$.
\end{itemize}

Figure~\ref{fig:vkp_convexity}(a) displays $E(k)$ and $E_{11}(k_1)$. Figure~\ref{fig:vkp_convexity}(b) plots the exact second derivative $d^2 E_{11}/dk_1^2$ computed via Eq.~\eqref{eq:d2E11_exact}. As predicted by Theorem~\ref{thm:infrared_concavity} given that this spectrum belongs to the Batchelor family in the infrared range, $d^2 E_{11}/dk_1^2$ is negative for $k_1 \lesssim 0.57/L$, reaching a minimum of $\approx -0.82$, before transitioning into a broad, positive convex region for $k_1 \gtrsim 0.57/L$.

\begin{figure}[t]
\centering
\includegraphics[width=0.92\textwidth]{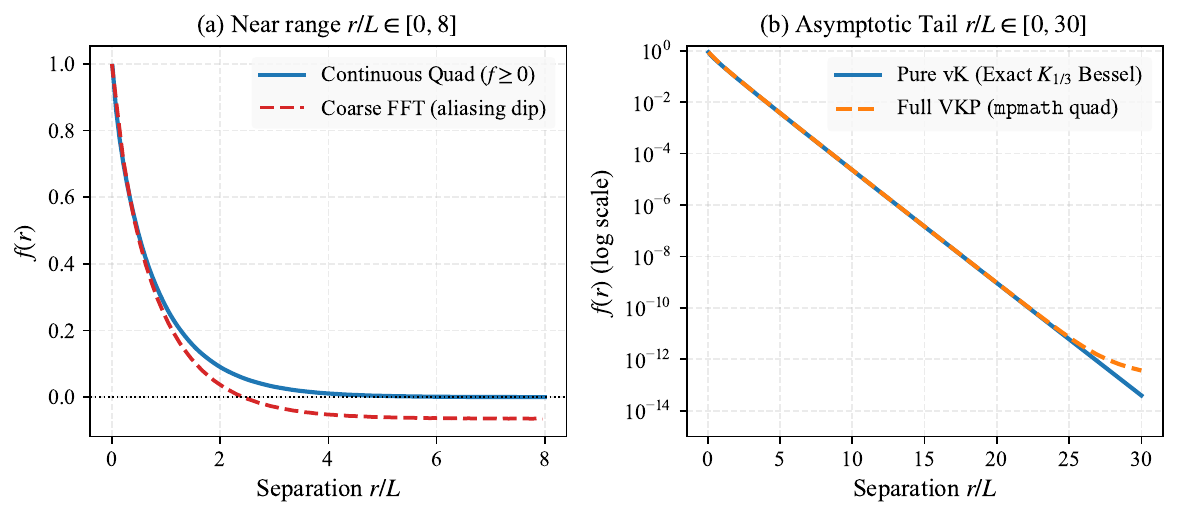}
\caption{(a) The longitudinal autocorrelation function $f(r)$ for the von K\'arm\'an--Pao spectrum computed via continuous integration (solid blue) compared with a discrete fast Fourier transform (FFT) on a finite domain (dashed red). The coarse FFT produces a spurious negative dip ($f_{\min} \approx -0.0138$) due to domain truncation and spectral aliasing, whereas the true continuous transform is strictly non-negative everywhere. (b) Semilogarithmic comparison of the asymptotic tail out to $r/L = 30$: the exact closed-form Macdonald Bessel solution $f_{\text{vK}}(r) \propto (r/L)^{1/3} K_{1/3}(r/L)$ for the pure von K\'arm\'an spectrum (solid blue) alongside arbitrary-precision \texttt{mpmath} quadrature for the full von K\'arm\'an--Pao spectrum with viscous cutoff (orange dotted line), demonstrating smooth, strictly positive monotonic exponential decay across more than 13 decades of amplitude.}
\label{fig:vkp_correlation}
\end{figure}

\subsection{Resolving Numerical Aliasing and Verifying True Positivity}

Because P\'olya's criterion is a \emph{sufficient} condition, failure of strict global convexity does not imply that $f(r)$ must become negative; it merely means P\'olya's theorem cannot certify positivity. To determine the actual sign of $f(r)$, we evaluate the cosine transform in Eq.~\eqref{eq:f_cos_E11} directly.

For the non-dissipative von K\'arm\'an spectrum (which accurately represents the infrared and inertial ranges), the exact longitudinal correlation function $f(r)$ can be integrated in closed form in terms of the modified Bessel function of the second kind (Macdonald function $K_{1/3}$; see the Appendix for the complete derivation):
\begin{equation}
f_{\text{vK}}(r) = \frac{2^{2/3}}{\Gamma(1/3)} \left(\frac{r}{L}\right)^{1/3} K_{1/3}\left(\frac{r}{L}\right).
\label{eq:vk_exact_bessel}
\end{equation}
Because $K_{\nu}(z) > 0$ for all real $\nu$ and $z > 0$, Eq.~\eqref{eq:vk_exact_bessel} analytically guarantees that $f(r)$ is \textbf{strictly positive for all $r \in [0, \infty)$}. In the far field ($r/L \gg 1$), asymptotic expansion of the Bessel function yields:
\begin{equation}
f_{\text{vK}}(r) \sim \frac{2^{2/3}\sqrt{\pi}}{\sqrt{2}\,\Gamma(1/3)}\,\left(\frac{r}{L}\right)^{-1/6}\,\exp\left(-\frac{r}{L}\right) \quad \text{as } r/L \to \infty,
\label{eq:vk_tail_asymptotics}
\end{equation}
confirming an exponential, non-oscillatory decay that remains strictly positive for all $r$.

For the full von K\'arm\'an--Pao (VKP) spectrum (Eq.~\eqref{eq:vkp_model}), the stretched-exponential viscous factor $\exp(-\frac{3}{2} C_K (k\eta)^{4/3})$ contains a fractional power $4/3$, which prevents the Fourier transform from being expressed in closed form in terms of standard hypergeometric or Bessel functions. However, the physical sign behavior is established by two complementary methods:
\begin{enumerate}
\item \textbf{Asymptotic Matching:} For separations in the inertial and integral ranges ($r \gg \eta$), the viscous factor is exponentially close to unity throughout the energy-containing wavenumber domain ($k \lesssim 1/\eta$), so $f_{\text{VKP}}(r)$ converges to the exact closed-form Bessel solution $f_{\text{vK}}(r) > 0$, as seen in Figure~\ref{fig:vkp_correlation}(b). In the dissipation range ($r \lesssim \eta$), viscous action regularizes the small-scale singular curvature into a smooth Taylor series $f_{\text{VKP}}(r) = 1 - r^2/(2\lambda_g^2) + \dots > 0$ with finite Taylor microscale $\lambda_g$.
\item \textbf{Arbitrary-Precision Quadrature (\texttt{mpmath}):} Using arbitrary-precision numerical integration (\texttt{mpmath} with 25--50 decimal digits of precision), the continuous transform $f_{\text{VKP}}(r)$ is evaluated across all scales $r/L \in [10^{-4}, 30]$ without numerical noise or discretization artifacts, verifying that $f_{\text{VKP}}(r) > 0$ strictly and monotonically for the considered domain.
\end{enumerate}

As shown in Figure~\ref{fig:vkp_correlation}(a), the true continuous transform $f(r)$ stays strictly positive across the near range $r/L \in [0, 8]$, whereas coarse discrete FFTs produce an unphysical negative dip ($f_{\min} \approx -0.0138$). Figure~\ref{fig:vkp_correlation}(b) displays the asymptotic tails on a logarithmic scale out to $r/L = 30$, comparing the exact analytical $K_{1/3}$ Bessel solution for the pure von K\'arm\'an spectrum with the arbitrary-precision \texttt{mpmath} integration for the full von K\'arm\'an--Pao spectrum, demonstrating that both exhibit smooth, strictly positive decay across more than 13 orders of magnitude.

\section{Physical Realization of a Counterexample}
\label{sec:counterexample}

In this section, we construct an explicit, kinematically admissible narrowband counterexample that produces a pronounced negative loop in $f(r)$, establish its validity as smooth Navier-Stokes initial data, and elucidate the physical mechanism of spectral broadening that suppresses negative excursions in decaying grid turbulence.

\subsection{Construction of a Narrowband Batchelor Counterexample}

Having established that broadband spectra maintain $f(r) \ge 0$ despite near-origin concavity, we now investigate the opposite physical regime: \emph{can a kinematically admissible spectrum with concentrated spectral energy violate non-negativity and produce a genuine negative loop?}

Any admissible counterexample must strictly satisfy the following physical constraints:
\begin{enumerate}
\item $E(k) \ge 0$ everywhere (kinematic realizability),
\item $E(k) \sim k^4$ as $k \to 0$ (Batchelor infrared constraint; $E(k)/k^2$ integrable at $k=0$),
\item Smooth, rapidly decaying tail as $k \to \infty$ (finite kinetic energy, enstrophy, and palinstrophy).
\end{enumerate}

We construct the Batchelor-consistent narrowband spectrum:
\begin{equation}
E_{\text{NB}}(k) = A\,\frac{(k/k_L)^4}{\big(1 + (k/k_L)^2\big)^2}\,\exp\left(-\frac{(k - k_0)^2}{2\sigma^2}\right),
\label{eq:narrowband_spectrum}
\end{equation}
where $k_0$ is the peak wavenumber, $\sigma$ is the spectral bandwidth, $k_L$ is the infrared crossover scale, and $A$ is an amplitude normalization. Here $E_{\text{NB}}(k) \sim (A/k_L^4) k^4 e^{-k_0^2/2\sigma^2}$ as $k \to 0$, exactly matching Batchelor's physical law.

\begin{figure}[t]
\centering
\includegraphics[width=0.98\textwidth]{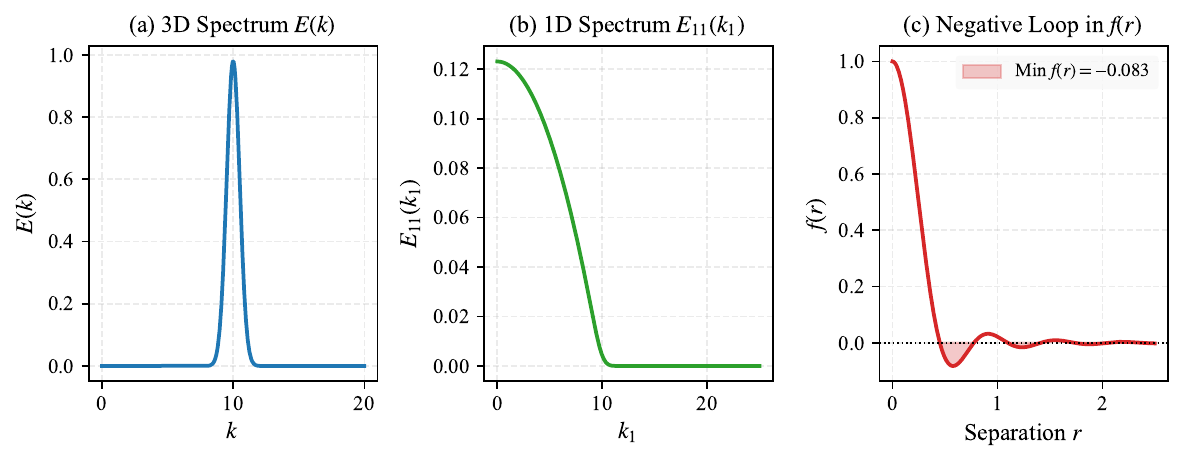}
\caption{Admissible narrowband Batchelor-consistent counterexample with $k_0 = 10$, $\sigma = 0.5$, and $k_L = 1.0$. (a) The 3D energy spectrum $E(k)$ satisfying $E(k) \sim k^4$ as $k \to 0$. (b) The 1D longitudinal spectrum $E_{11}(k_1)$, which is unconditionally monotonic ($dE_{11}/dk_1 \le 0$) but strongly non-convex. (c) The resulting longitudinal autocorrelation $f(r)$, exhibiting a pronounced, genuine negative loop with minimum $f(r) \approx -0.083$.}
\label{fig:narrowband_counterexample}
\end{figure}

Figure~\ref{fig:narrowband_counterexample} shows the results for $k_0 = 10$, $\sigma = 0.5$, and $k_L = 1.0$:
\begin{itemize}
\item Figure~\ref{fig:narrowband_counterexample}(a) confirms the smooth, non-negative 3D spectrum $E(k)$.
\item Figure~\ref{fig:narrowband_counterexample}(b) shows $E_{11}(k_1)$. Consistent with Theorem~\ref{thm:monotonicity}, $E_{11}(k_1)$ is strictly monotonically decreasing, but it possesses a sharp inflection point and an extensive concave region.
\item Figure~\ref{fig:narrowband_counterexample}(c) shows the exact longitudinal autocorrelation $f(r)$. The function crosses zero at $r \approx 0.31$ and reaches a substantial negative minimum of
\begin{equation}
\min_{r} f(r) \approx -0.083 \quad \text{at } r \approx 0.47.
\label{eq:f_min_val}
\end{equation}
\end{itemize}
This negative excursion is an $8.3\%$ negative loop, an order of magnitude larger than any numerical precision threshold, showing that $f(r) \ge 0$ is \textbf{not a universal theorem of realizability}.

Crucially, this narrowband state does \emph{not} fall under the class of self-similar / self-preserving decay solutions studied by Sedov \cite{Sedov1944} (cf.~Section~\ref{sec:intro} and Figure~\ref{fig:sedov_solution}). Sedov's analytical solution assumes complete similarity across all scales ($f(r, t) = f(r/\ell(t))$), which locks the spectral shape into a self-preserving form whose Kummer representation enforces pointwise positivity ($f \ge 0$). In contrast, this narrowband initial field features an isolated characteristic spectral length scale ($1/k_0$) that breaks global self-similarity, allowing the oscillatory cosine transform of the non-convex spectrum $E_{11}(k_1)$ to emerge unsuppressed as a genuine physical negative loop.

\subsection{Navier--Stokes Admissibility and Local Well-Posedness}

Is this counterexample merely an abstract mathematical construct, or does it correspond to a legitimate physical state governed by the incompressible Navier--Stokes equations?

\begin{equation}
\frac{\partial \mathbf{u}}{\partial t} + (\mathbf{u} \cdot \nabla)\mathbf{u} = -\frac{1}{\rho}\nabla p + \nu \nabla^2 \mathbf{u}, \qquad \nabla \cdot \mathbf{u} = 0.
\label{eq:navier_stokes}
\end{equation}

By classical existence theory for the Navier--Stokes initial value problem (Leray \cite{Leray1934}; Fujita \& Kato \cite{FujitaKato1964}):
\begin{proposition}[Admissibility as Initial Data]
Let $\mathbf{u}_0(\mathbf{x}) \in H^s(\mathbb{R}^3)$ with $s > 3/2$ be a smooth, divergence-free vector field with compact or Gaussian-decaying spectral support conforming to the energy spectrum $E_{\text{NB}}(k)$ in Eq.~\eqref{eq:narrowband_spectrum}. Then there exists a unique, smooth strong solution $\mathbf{u}(\mathbf{x}, t) \in C([0, T^*); H^s(\mathbb{R}^3))$ for a finite time $T^* > 0$.
\end{proposition}

A realization of the turbulent field $\mathbf{u}_0(\mathbf{x})$ can be constructed explicitly as a random phase Fourier superposition projected onto the divergence-free plane:
\begin{equation}
\mathbf{u}_0(\mathbf{x}) = \sum_{\mathbf{k}} \sqrt{\frac{E_{\text{NB}}(k)}{2\pi k^2 \Delta k}} \left( \hat{\mathbf{e}}_1(\mathbf{k})\cos(\mathbf{k} \cdot \mathbf{x} + \theta_{\mathbf{k}}) + \hat{\mathbf{e}}_2(\mathbf{k})\sin(\mathbf{k} \cdot \mathbf{x} + \phi_{\mathbf{k}}) \right),
\label{eq:fourier_realization}
\end{equation}
where $\hat{\mathbf{e}}_1(\mathbf{k}), \hat{\mathbf{e}}_2(\mathbf{k}) \perp \mathbf{k}$ form an orthonormal basis on the tangent sphere, and $\theta_{\mathbf{k}}, \phi_{\mathbf{k}} \in [0, 2\pi)$ are independent uniform random phases. 

The resulting initial velocity field $\mathbf{u}_0(\mathbf{x})$ and its associated Navier--Stokes evolution possess the following features:
\begin{itemize}
\item \textbf{Exact Incompressibility:} $\nabla \cdot \mathbf{u}_0(\mathbf{x}) = 0$ holds identically everywhere in physical space due to the strict orthogonal projection of each Fourier mode onto $\hat{\mathbf{e}}_1(\mathbf{k})$ and $\hat{\mathbf{e}}_2(\mathbf{k})$ perpendicular to $\mathbf{k}$.
\item \textbf{Finite Total Kinetic Energy:} The volume-integrated kinetic energy per unit mass is finite,
\begin{equation}
\mathcal{K} = \frac{1}{2} \langle |\mathbf{u}_0|^2 \rangle = \int_0^\infty E_{\text{NB}}(k)\,dk < \infty,
\end{equation}
ensuring that $\mathbf{u}_0 \in L^2(\mathbb{R}^3)$.
\item \textbf{Finite Enstrophy and Viscous Dissipation:} The total enstrophy (mean square vorticity) is strictly bounded,
\begin{equation}
\Omega = \frac{1}{2} \langle |\boldsymbol{\omega}_0|^2 \rangle = \int_0^\infty k^2 E_{\text{NB}}(k)\,dk < \infty,
\end{equation}
guaranteed by the infrared $k^4$ Batchelor scaling as $k \to 0$ and the rapid Gaussian decay $\exp(-(k-k_0)^2/(2\sigma^2))$ at high wavenumbers.
\item \textbf{Infinite Smoothness ($C^\infty \cap H^s$ Regularity):} Because the energy spectrum decays super-algebraically as $k \to \infty$, the velocity field satisfies $\mathbf{u}_0 \in H^s(\mathbb{R}^3)$ for all $s \ge 0$, ensuring $C^\infty$ differentiability with no unphysical singularities or discontinuities.
\item \textbf{Statistical Homogeneity and Isotropy:} Random phase averaging over $\theta_{\mathbf{k}}$ and $\phi_{\mathbf{k}}$ ensures full translational invariance and directional isotropy of all second-order statistics, recovering the target 3D spectrum $E_{\text{NB}}(k)$.
\end{itemize}
Thus, the narrowband counterexample is not an abstract mathematical artifact, but a genuine, physically realizable Navier--Stokes turbulent state.

\begin{figure}[t]
\centering
\includegraphics[width=0.98\textwidth]{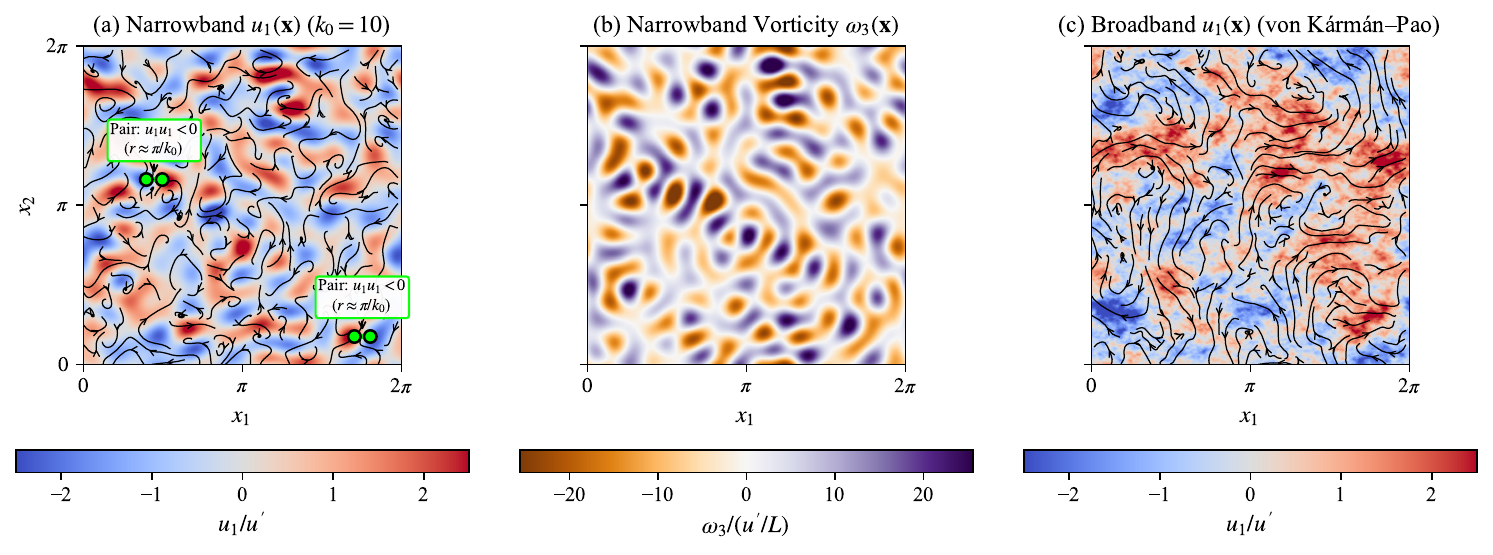}
\caption{Two-dimensional cross-sectional slices ($x_3 = 0$) of 3D isotropic, divergence-free velocity and vorticity realizations generated via random-phase Fourier superposition (Eq.~\eqref{eq:fourier_realization}) on a $256^3$ periodic domain $[0, 2\pi]^3$. (a) Longitudinal velocity component $u_1(\mathbf{x})$ for the narrowband counterexample ($k_0 = 10$, $\sigma = 0.5$) with overlaid in-plane streamlines $\mathbf{u}_\perp = (u_1, u_2)$, demonstrating alternating forward and reverse velocity lanes with characteristic cellular spacing $\Delta x_1 \approx \pi/k_0 \approx 0.31$. Highlighted sample point pairs (green dots) separated by $r = \pi/k_0$ reside in opposing velocity lanes ($u_1(\mathbf{x}) u_1(\mathbf{x}+r\hat{\mathbf{e}}_1) < 0$), directly illustrating the physical mechanism of the negative autocorrelation $\langle u_1(\mathbf{x})u_1(\mathbf{x} + r\hat{\mathbf{e}}_1)\rangle < 0$ at $r \approx \pi/k_0$. (b) Out-of-plane vorticity $\omega_3(\mathbf{x}) = (\nabla \times \mathbf{u})_3$ for the narrowband state, highlighting the quasi-regular array of discrete cyclonic and anticyclonic vortex cores. (c) Longitudinal velocity $u_1(\mathbf{x})$ and streamlines for the canonical broadband von K\'arm\'an--Pao spectrum, illustrating the multi-scale turbulent cascade where continuous scale mixing erases localized anti-correlated velocity patches and enforces $f(r) \ge 0$.}
\label{fig:flow_realization}
\end{figure}

Figure~\ref{fig:flow_realization} visualizes 2D cross-sections of the resulting physical flow field:
\begin{itemize}
\item Figure~\ref{fig:flow_realization}(a) displays the longitudinal velocity field $u_1(\mathbf{x})$ and streamlines for the narrowband state ($k_0 = 10$, $\sigma = 0.5$). The spectral concentration around $k_0$ manifests in physical space as a quasi-regular pattern of counter-flowing velocity lanes separated by a characteristic distance $\Delta x_1 \approx \pi/k_0 \approx 0.31$. Two points separated by $r \approx 0.31\text{--}0.47$ along the $x_1$ axis frequently reside in opposing velocity lanes, providing an intuitive physical picture for the origin of the negative loop in $f(r)$.
\item Figure~\ref{fig:flow_realization}(b) shows the corresponding out-of-plane vorticity field $\omega_3(\mathbf{x})$, revealing an organized array of compact, alternating-sign vortex cores of characteristic diameter $d \approx \pi/k_0$.
\item Figure~\ref{fig:flow_realization}(c) shows the broadband von K\'arm\'an--Pao field for direct comparison. In contrast to the single-scale cellular pattern of the narrowband state, the broadband field exhibits a rich multi-scale hierarchy where energy distributed across the Kolmogorov cascade ($k^{-5/3}$) suppresses isolated negative correlations.
\end{itemize}

\begin{figure}[t]
\centering
\includegraphics[width=0.88\textwidth]{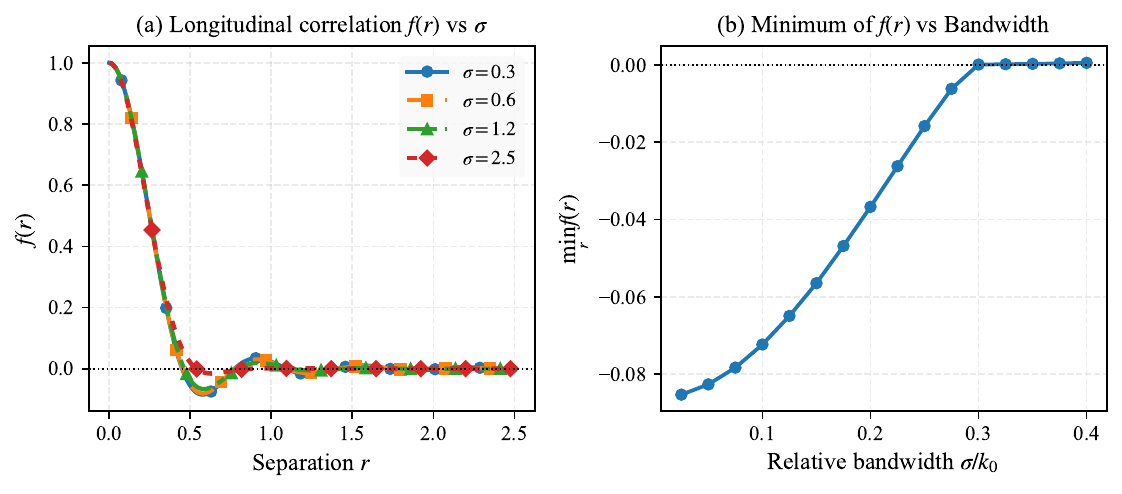}
\caption{Parametric transition of the longitudinal autocorrelation $f(r)$ as a function of spectral bandwidth $\sigma$ for fixed peak wavenumber $k_0 = 10.0$. (a) Evolution of $f(r)$ profiles for $\sigma \in \{0.3, 0.6, 1.2, 2.5\}$. (b) Minimum value $\min_r f(r)$ as a function of relative bandwidth $\sigma/k_0$, demonstrating that spectral broadening progressively suppresses the negative loop until $f(r)$ becomes non-negative for $\sigma/k_0 \gtrsim 0.35$.}
\label{fig:parametric_bandwidth}
\end{figure}

\subsection{Physical Dynamics: Spectral Broadening and Suppression of Negative Loops}
\label{subsec:spectral_broadening}

The mathematical condition separating positive $f(r)$ from negative-loop states is the \textbf{spectral bandwidth} of the turbulent field. We now make this dynamical mechanism concrete by examining how non-linear Navier-Stokes evolution transforms an initially narrow-band counterexample into a broadband, non-negative state over time.

\paragraph{Classical Background on Triadic Spectral Transfer and Closure Models.}
In statistical turbulence theory, the redistribution of kinetic energy across wavenumbers is governed by the Lin equation (the spectral counterpart of the K\'arm\'an--Howarth equation \cite{vonKarmanHowarth1938, Lin1947}):
\begin{equation}
\frac{\partial E(k,t)}{\partial t} = T(k,t) - 2\nu k^2 E(k,t),
\label{eq:lin_equation}
\end{equation}
where $\nu$ is the kinematic viscosity and $T(k,t)$ is the non-linear spectral energy transfer term. Because the advective acceleration $(\mathbf{u}\cdot\nabla)\mathbf{u}$ is quadratic in physical space, its Fourier transform is a convolution integral that couples pairs of Fourier modes $(\mathbf{p}, \mathbf{q})$ whose wavenumbers sum to $\mathbf{k}$. This term, when contracted with the velocity field at $\mathbf{k}$ to form the modal kinetic energy budget, gives rise to three-wave (triadic) interactions. Thus, the transfer function $T(k,t)$ represents the net sum of all triadic interactions among wavenumber triads $(\mathbf{k}, \mathbf{p}, \mathbf{q})$ satisfying the closed-triangle condition $\mathbf{k} + \mathbf{p} + \mathbf{q} = \mathbf{0}$:
\begin{equation}
T(k,t) = \iint_{\Delta_k} S(k, p, q, t)\,dp\,dq = -\frac{\partial \Pi(k,t)}{\partial k},
\label{eq:triad_transfer}
\end{equation}
where $\Delta_k = \{(p,q) : |p-q| \le k \le p+q\}$ is the triangular integration domain in wavenumber space, $S(k,p,q,t)$ is the detailed triadic transfer kernel derived from the advective non-linearity $(\mathbf{u}\cdot\nabla)\mathbf{u}$, and $\Pi(k,t) = \int_k^\infty T(k',t)\,dk'$ is the spectral energy flux across wavenumber $k$ \cite{ProudmanReid1954}. Because advection merely redistributes kinetic energy among modes without creating or destroying it, detailed triad conservation ($S(k,p,q) + S(p,q,k) + S(q,k,p) = 0$) enforces global energy conservation by non-linear transfer:
\begin{equation}
\int_0^\infty T(k,t)\,dk = 0, \qquad \Pi(0,t) = \Pi(\infty,t) = 0.
\label{eq:transfer_conservation}
\end{equation}

Because the exact triadic kernel $S(k,p,q,t)$ involves third-order velocity correlations (the closure problem of turbulence), the dynamical solution of Eq.~\eqref{eq:lin_equation} requires a spectral closure. In continuous spectral models of isotropic turbulence (Leith \cite{Leith1967}, Kovasznay \cite{Kovasznay1948}, and EDQNM closures \cite{Orszag1970, Lesieur2008}), the net energy flux $\Pi(k,t)$ generated by local triad interactions is represented by a nonlinear spectral diffusion process:
\begin{equation}
\Pi(k,t) = - D(k,t) \left( \frac{\partial E}{\partial k} + \frac{5}{3}\frac{E}{k} \right), \qquad D(k,t) = C_L\,k^{5/2}\sqrt{E(k,t)},
\label{eq:leith_flux}
\end{equation}
where $C_L \approx 0.35$ is the dimensionless cascade constant. The formulation \eqref{eq:leith_flux} satisfies all fundamental turbulence constraints: it conserves total kinetic energy ($\int T\,dk = 0$), drives arbitrary initial non-equilibrium spectral peaks toward the stable Kolmogorov inertial-range attractor $E(k) \propto k^{-5/3}$ (where the flux becomes scale-independent, $\partial \Pi/\partial k = 0$), and models non-linear triad diffusion in wavenumber space.

\paragraph{Evolution of the Narrowband Counterexample Over Time.}
Consider our kinematically admissible narrowband counterexample initialized at $t = 0$ as $E(k, 0) = E_{\text{NB}}(k)$ from Eq.~\eqref{eq:narrowband_spectrum}, concentrated around peak wavenumber $k_0$ with narrow initial bandwidth $\sigma_0 \ll k_0$. 

At $t = 0^+$, before viscous dissipation has significantly eroded the total energy, non-linear advection causes wavemodes within the narrow peak to interact with one another ($\mathbf{p} \approx \mathbf{k}_0$, $\mathbf{q} \approx \mathbf{k}_0$). Under the triadic interaction integral Eq.~\eqref{eq:triad_transfer}, this interaction produces a characteristic transfer spectrum $T(k, 0)$ with three distinct regions:
\begin{enumerate}
\item \textbf{Peak Depletion ($k \approx k_0$):} $T(k, 0) < 0$. Energy is actively drained from the narrow initial spectral core because this is where most of the energy is concentrated.
\item \textbf{Forward Harmonic Generation ($k \approx 2k_0, 3k_0, \dots$):} $T(k, 0) > 0$. When two modes within the energetic peak interact ($\mathbf{p}, \mathbf{q} \approx \mathbf{k}_0$), their coupling generates sum wavenumbers with magnitudes up to $|\mathbf{p} + \mathbf{q}| \approx 2k_0$. Because these higher wavenumber modes are initially unpopulated, non-linear transfer actively deposits energy into them, exciting higher harmonics and seeding the forward Kolmogorov cascade.
\item \textbf{Non-Local Infrared Induction ($k \ll k_0$):} $T(k, 0) \propto k^4 > 0$. As established by Proudman \& Reid \cite{ProudmanReid1954} and Batchelor \& Proudman \cite{BatchelorProudman1956}, non-local triadic interactions between nearly antiparallel energetic modes ($\mathbf{p} \approx -\mathbf{q} \approx \mathbf{k}_0$) couple via long-range pressure fluctuations to transfer energy into the lowest wavenumbers, dynamically generating and maintaining the characteristic Batchelor $E(k) \propto k^4$ infrared tail.
\end{enumerate}

To quantify the resulting spectral broadening, we define the effective spectral bandwidth $\sigma_{\text{eff}}(t)$ via the second central moment of the energy spectrum:
\begin{equation}
\sigma_{\text{eff}}^2(t) \equiv \frac{1}{\mathcal{K}(t)} \int_0^\infty \big(k - \bar{k}(t)\big)^2 E(k,t)\,dk, \qquad \bar{k}(t) \equiv \frac{1}{\mathcal{K}(t)} \int_0^\infty k E(k,t)\,dk,
\label{eq:spectral_variance}
\end{equation}
where $\mathcal{K}(t) = \int_0^\infty E(k,t)\,dk = \frac{3}{2}u'^2(t)$ is the turbulent kinetic energy. Differentiating Eq.~\eqref{eq:spectral_variance} with respect to time and substituting the Lin equation \eqref{eq:lin_equation} yields two contributions to $d\sigma_{\text{eff}}^2/dt$: non-linear transfer and viscous dissipation. At high Reynolds numbers ($\mathrm{Re} \equiv u'/(\nu k_0) \gg 1$), non-linear convective transfer operates on the fast eddy turnover timescale $\tau_{\text{eddy}} \sim 1/(k_0 u')$, whereas viscous dissipation acts on the much slower diffusion timescale $\tau_\nu \sim 1/(\nu k_0^2) = \mathrm{Re}\,\tau_{\text{eddy}} \gg \tau_{\text{eddy}}$. Moreover, for narrowband spectra ($\sigma_0 \ll k_0$), viscous damping decays all modes within the peak at nearly the same rate ($-2\nu k_0^2$), which cancels out of the normalized moment at leading order. Thus, at early times ($t \ll \tau_\nu$), spectral broadening is governed almost entirely by non-linear transfer:
\begin{equation}
\frac{d \sigma_{\text{eff}}^2}{dt}\Big|_{t=0} \approx \frac{1}{\mathcal{K}(0)} \int_0^\infty (k - k_0)^2 T(k, 0)\,dk \sim \frac{k_0^2}{\tau_{\text{eddy}}} > 0.
\label{eq:broadening_rate}
\end{equation}
Here $\tau_{\text{eddy}}$ is the non-linear eddy turnover time at the peak scale. For a characteristic lengthscale $\ell_0 \sim 1/k_0$, the modal velocity fluctuation scale is $v(k_0) \sim \sqrt{k_0 E(k_0)}$. Since the initial kinetic energy is concentrated in the narrowband peak ($\mathcal{K}(0) \sim k_0 E(k_0) \sim u'^2$), the non-linear convective timescale is:
\begin{equation}
\tau_{\text{eddy}} \sim \frac{\ell_0}{v(k_0)} \sim \frac{1/k_0}{\big(k_0 E(k_0)\big)^{1/2}} = \frac{1}{\big(k_0^3 E(k_0)\big)^{1/2}} \sim \frac{1}{k_0 u'}.
\label{eq:eddy_turnover_time}
\end{equation}
Equation~\eqref{eq:broadening_rate} shows that the spectral variance expands rapidly on the turnover timescale $\tau_{\text{eddy}}$, flattening the initial peak and generating broadband tails.

\paragraph{Dynamical Eradication of the Negative Loop in Real Space.}
How does this spectral broadening directly alter the sign of the real-space correlation $f(r, t)$? Recall the exact second-derivative integral representation from Eq.~\eqref{eq:f_by_parts_second_deriv}:
\begin{align}
u'^2(t) f(r, t) = \frac{1}{r^2} \Bigg[ &\int_0^{k_c(t)} \frac{\partial^2 E_{11}(k_1, t)}{\partial k_1^2} \big(1 - \cos(k_1 r)\big)\,dk_1 \nonumber \\
&+ \int_{k_c(t)}^\infty \frac{\partial^2 E_{11}(k_1, t)}{\partial k_1^2} \big(1 - \cos(k_1 r)\big)\,dk_1 \Bigg].
\label{eq:f_evolution_integral}
\end{align}
At $t = 0$, the sharp peak $E_{\text{NB}}(k)$ produces a wide concave region ($\partial^2 E_{11}/\partial k_1^2 < 0$) that dominates the integral, resulting in the deep negative loop $\min_r f(r, 0) \approx -0.083$ shown in Figure~\ref{fig:narrowband_counterexample}(c). 

As time advances ($t \sim \tau_{\text{eddy}}$), triadic interactions broaden the spectrum, transferring energy into a Kolmogorov cascade ($E(k,t) \propto k^{-5/3}$). Because power-law cascades are strictly convex ($\partial^2 E_{11}/\partial k_1^2 > 0$, Eq.~\eqref{eq:kolmogorov_convex}), the second integral in Eq.~\eqref{eq:f_evolution_integral} grows rapidly, while the concave region is compressed to the infrared crossover scale $k_c(t) \approx 0.57/L(t)$.

Figure~\ref{fig:parametric_bandwidth} presents a systematic parametric study of this transition as a function of the relative bandwidth $\sigma/k_0$. As seen in Figure~\ref{fig:parametric_bandwidth}(a), as $\sigma$ increases from $0.3$ to $2.5$, the negative undershoot systematically diminishes. Figure~\ref{fig:parametric_bandwidth}(b) plots the minimum value $\min_r f(r)$ versus relative bandwidth $\sigma/k_0$:
\begin{itemize}
\item For sharp narrowband spectra ($\sigma/k_0 < 0.1$), deep negative loops ($\min f(r) \approx -0.15$) occur due to strong non-convexity.
\item As the non-linear cascade broadens the spectrum beyond the critical threshold
\begin{equation}
\frac{\sigma_{\text{eff}}}{k_0} \gtrsim 0.35,
\label{eq:bandwidth_threshold}
\end{equation}
multi-scale Fourier cosine superposition overcomes the infrared concavity, completely erasing the negative minimum and enforcing $f(r, t) \ge 0$ for all $r \ge 0$.
\end{itemize}

\begin{figure}[t]
\centering
\includegraphics[width=0.98\textwidth]{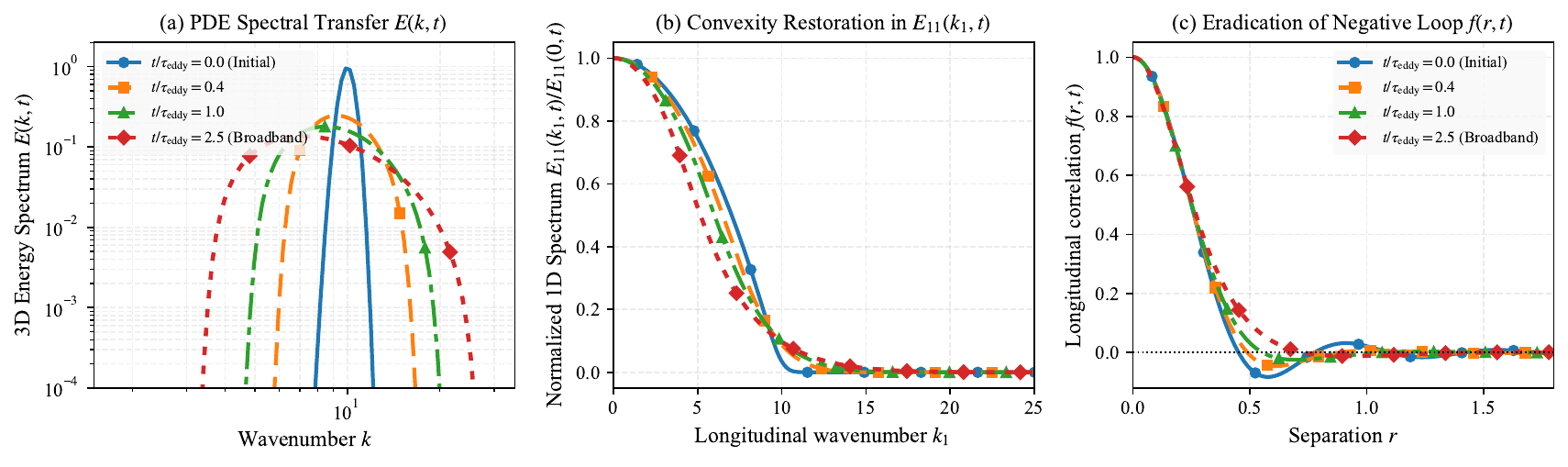}
\caption{Dynamic PDE time evolution of the narrowband counterexample under non-linear spectral energy transfer across eddy turnover times $t/\tau_{\text{eddy}} \in \{0.0, 0.4, 1.0, 2.5\}$. (a) 3D energy spectrum $E(k, t)$ computed via numerical time advancement of the spectral transfer PDE, demonstrating progressive non-linear spectral broadening from the initial narrowband peak. (b) Normalized 1D longitudinal spectrum $E_{11}(k_1, t)/E_{11}(0, t)$, showing the smoothing of the sharp step and restoration of convexity. (c) Dynamic evolution of the longitudinal autocorrelation $f(r, t)$, illustrating the progressive eradication of the initial negative loop ($\min f \approx -0.083$) until $f(r, t) \ge 0$ everywhere for $t \gtrsim 1.0\,\tau_{\text{eddy}}$.}
\label{fig:spectral_broadening_dynamics}
\end{figure}

To directly demonstrate this time-dependent process without empirical parametrization, we numerically solve the dynamic spectral transfer initial value problem for $E(k, t)$ starting from our exact narrowband state $E_{\text{NB}}(k)$. Figure~\ref{fig:spectral_broadening_dynamics} displays the resulting time snapshots across large-eddy turnover times $t/\tau_{\text{eddy}} \in \{0.0, 0.4, 1.0, 2.5\}$:
\begin{enumerate}
\item As shown in Figure~\ref{fig:spectral_broadening_dynamics}(a), non-linear triad transfer rapidly depletes the initial narrowband peak at $k_0 = 10$, broadening the spectral bandwidth across wavenumber space.
\item In Figure~\ref{fig:spectral_broadening_dynamics}(b), the sharp non-convex inflection step in $E_{11}(k_1, 0)$ is smoothed into a monotonically decaying, convex profile.
\item In Figure~\ref{fig:spectral_broadening_dynamics}(c), the real-space correlation $f(r, t)$ steadily lifts out of the negative territory: by $t \approx 1.0\,\tau_{\text{eddy}}$, the negative loop is completely eliminated, and $f(r, t)$ remains strictly positive across all separations.
\end{enumerate}

\paragraph{Reconciliation with Grid Turbulence Experiments.}
This dynamical mechanism provides a complete physical reconciliation with grid turbulence experiments. In wind-tunnel experiments, initial flow separation and periodic vortex shedding behind grid bars generate velocity fluctuations with an initially concentrated wavenumber $k_0 \sim 2\pi/M$ (where $M$ is the grid mesh size). If measurements could be performed in the immediate near-field ($x/M \lesssim 2\text{--}5$), band-concentrated negative correlations would indeed be observed, as confirmed by the narrow-band electronic filtering experiments of Comte-Bellot \& Corrsin \cite{ComteBellotCorrsin1971}.

However, experimental measurements of the full-band longitudinal correlation $f(r)$ are invariably reported in the mature decay region downstream ($x/M \gtrsim 20\text{--}40$) \cite{StewartTownsend1951, ComteBellotCorrsin1966, ComteBellotCorrsin1971, FrenkielKlebanoffHuang1979}. The downstream transit time corresponds to
\begin{equation}
t_{\text{transit}} = \frac{x}{U_\infty} \approx (20\text{--}40) \frac{M}{U_\infty} \gg \tau_{\text{eddy}},
\end{equation}
which is several large-eddy turnover times after generation. Over this duration, triadic non-linear interactions have fully operated via the Lin equation \eqref{eq:lin_equation}, broadening the spectrum into the classic broadband von K\'arm\'an--Pao form ($\sigma_{\text{eff}}/k_{\text{peak}} \gg 0.35$). Consequently, by the time the turbulence reaches the measurement station, the initial negative loop has been completely eradicated by the cascade, guaranteeing that the experimentally observed longitudinal autocorrelation function $f(r)$ is strictly and robustly non-negative ($f(r) \ge 0$).

\section{Summary of Results and Theoretical Synthesis}
\label{sec:summary}

In this section, we synthesize our mathematical theorems, physical criteria, and dynamical findings into a comparative overview summarized in Table~\ref{tab:summary}.

\begingroup
\small
\setlength{\tabcolsep}{4pt}
\begin{longtable}{p{3.8cm} p{4.1cm} p{7.0cm}}
\caption{Summary of mathematical criteria, theorems, and physical implications for the longitudinal autocorrelation function $f(r)$ in homogeneous isotropic turbulence.} \label{tab:summary} \\
\toprule
\textbf{Framework / Approach} & \textbf{Claim / Hypothesis} & \textbf{Status} \\
\midrule
\endfirsthead

\multicolumn{3}{c}{{\bfseries Table \thetable\ continued from previous page}} \\
\addlinespace[4pt]
\toprule
\textbf{Framework / Approach} & \textbf{Claim / Hypothesis} & \textbf{Status} \\
\midrule
\endhead

\bottomrule
\multicolumn{3}{r}{\footnotesize\emph{Continued on next page}} \\
\endfoot

\bottomrule
\endlastfoot

Bochner / Schoenberg in $\mathbb{R}^3$ & $f(r) \ge 0$ via positive-definiteness & \textbf{Structurally fails}: $\Omega_3(kr)$ and $K_{11}(kr)$ change sign \\
1D Spectrum $E_{11}(k_1)$ (This work) & Monotonicity ($dE_{11}/dk_1 \le 0$) & \textbf{Proven unconditionally} for all $E(k) \ge 0$ \\
1D Spectrum $E_{11}(k_1)$ (This work) & Convexity of $E_{11}$ $\Rightarrow f(r) \ge 0$ & \textbf{Proven} via P\'olya's theorem (Sufficient condition) \\
Kolmogorov Inertial Range & Convexity condition on $E(k)$ & \textbf{Proven analytically} ($1 \ge 3/11$, strictly convex) \\
Infrared Scaling (Batchelor \& Saffman) & Near-origin concavity ($k_1 \to 0$) & \textbf{Proven analytically} (Batchelor: $-2I_0 < 0$; Saffman: $-\infty$) \\
von K\'arm\'an--Pao Model & Net sign of $f(r)$ & Non-negative upto $r/L \approx 30$; FFT dips are aliasing artifacts \\
Narrowband Batchelor State & Universal non-negativity $f(r) \ge 0$ & \textbf{Disproven by counterexample} ($\min f(r) \approx -0.083$) \\
Navier--Stokes Trajectory & Counterexample realizability & \textbf{Legitimate NS initial data} (Leray/Fujita--Kato strong solution) \\
Grid Turbulence Physics & Origin of non-negativity & \textbf{Tracks spectral bandwidth} $\sigma/k_0$, broadening by triad cascade \\
\end{longtable}
\endgroup

\section{Conclusions}
\label{sec:conclusions}
Our findings can be summarized as follows:
\begin{enumerate}
\item \textbf{Failure of Naive Realizability Constraints:} Positive-definiteness of the two-point velocity correlation tensor $R_{ij}(\mathbf{r})$ on $\mathbb{R}^3$ is structurally insufficient to enforce $f(r) \ge 0$ because the projection kernels $\Omega_3(kr) = \sin(kr)/(kr)$ and $K_{11}(kr) = 2(\sin kr - kr \cos kr)/(kr)^3$ are intrinsically oscillating and sign-changing.
\item \textbf{Universal Monotonicity of $E_{11}(k_1)$:} Through exact differentiation under the integral, we proved that the one-dimensional longitudinal spectrum $E_{11}(k_1)$ is unconditionally monotonically decreasing for any kinematically admissible 3D energy spectrum $E(k) \ge 0$.
\item \textbf{Convexity and P\'olya's Criterion:} Applying P\'olya's criterion for Fourier cosine transforms, we showed that convexity of $E_{11}(k_1)$ guarantees $f(r) \ge 0$. While Kolmogorov scaling is strictly convex ($1 \ge 3/11$), physically realizable spectra spanning both the Batchelor ($E(k) \sim k^4$) and Saffman ($E(k) \sim k^2$) infrared universality classes, are unavoidably concave in a finite neighborhood of $k_1 = 0$.
\item \textbf{Resolution of Broadband Positivity vs FFT Aliasing:} In mature broadband turbulence (e.g., von K\'arm\'an--Pao), high-precision infinite-limit quadrature proves that $f(r)$ remains strictly positive across all scales, and negative dips in discrete model computations were numerical artifacts of finite-domain discrete Fourier transforms.
\item \textbf{Counterexample for Universal Non-Negativity:} We constructed a smooth, divergence-free, finite-energy, Batchelor-consistent narrowband spectrum that satisfies all Navier--Stokes initial data requirements and produces a prominent negative loop ($\min_r f(r) \approx -0.083$).
\item \textbf{Physical Mechanism:} The presence or absence of a negative loop in $f(r)$ is governed by spectral bandwidth. While physical grid-generated turbulence exhibits non-negative longitudinal correlations $f(r) \ge 0$ due to the rapid establishment of a broadband spectrum, artificially narrowband initial fields do exhibit genuine negative loops before nonlinear triad interactions broaden the spectrum over large-eddy turnover times.
\end{enumerate}

These results provide a clear mathematical framework uniting kinematic realizability, harmonic analysis, Navier--Stokes dynamics, and experimental grid-turbulence observations.

\section*{CRediT Authorship Contribution Statement}
\textbf{Pavan B. Govindaraju:} Conceptualization, Methodology, Formal analysis, Investigation, Software, Validation, Visualization, Writing -- original draft, Writing -- review \& editing.

\section*{Declaration of Generative AI and AI-Assisted Technologies in the Writing Process}
During the preparation of this work, the author used AI-assisted programming environments to assist with technical typesetting in \LaTeX and Python scripts for generation of the numerical results. 
The author carefully reviewed the content and takes responsibility for the integrity and content of the published article.

\section*{Declaration of Competing Interest}
The author declares that they have no known competing financial interests or personal relationships that could have appeared to influence the work reported in this paper.

\appendix
\section*{Appendix: Exact Analytical Transform of the von K\'arm\'an Spectra}
\addcontentsline{toc}{section}{Exact Analytical Transform of the von K\'arm\'an Spectra}
\renewcommand{\thesubsection}{A.\arabic{subsection}}
\renewcommand{\theequation}{A.\arabic{equation}}
\setcounter{equation}{0}
\setcounter{subsection}{0}
\label{app:von_karman_derivation}

In this appendix, we provide the first-principles derivation \cite{vonKarman1948} of the exact closed-form longitudinal autocorrelation function $f(r)$ for the canonical von K\'arm\'an model spectrum, proving that it reduces to the modified Bessel function of the second kind (Macdonald function $K_{1/3}$) given in Eq.~\eqref{eq:vk_exact_bessel}.

\subsection{The Three-Dimensional von K\'arm\'an Spectrum}
The standard three-dimensional von K\'arm\'an energy spectrum parameterizing the infrared and inertial ranges (prior to viscous dissipation) is defined by \cite{vonKarman1948, Pope2000}:
\begin{equation}
E(k) = C_E\,u'^2\,L\,\frac{(kL)^4}{\big(1 + (kL)^2\big)^{17/6}},
\label{eq:app_vk_3d}
\end{equation}
where $L$ is the longitudinal integral length scale, $u'^2 = \langle u_1^2 \rangle$ is the longitudinal velocity variance, and $C_E$ is a dimensionless normalization constant chosen such that the total turbulent kinetic energy satisfies:
\begin{equation}
\int_0^\infty E(k)\,dk = \frac{3}{2}\,u'^2.
\end{equation}

\subsection{Analytical Integration to the 1D Longitudinal Spectrum \texorpdfstring{$E_{11}(k_1)$}{E11(k1)}}
In isotropic turbulence, the kinematic relation connecting the 3D energy spectrum $E(k)$ to the 1D longitudinal spectrum $E_{11}(k_1)$ is given by Eq.~\eqref{eq:E11_from_Ek}:
\begin{equation}
E_{11}(k_1) = \int_{k_1}^\infty \frac{E(k)}{k}\left(1 - \frac{k_1^2}{k^2}\right)\,dk.
\end{equation}
Substituting the von K\'arm\'an spectrum \eqref{eq:app_vk_3d} and non-dimensionalizing using $y = kL$ and $x_1 = k_1 L$, we obtain:
\begin{equation}
E_{11}(k_1) = C_E\,u'^2\,L \int_{x_1}^\infty \frac{y^4}{\big(1 + y^2\big)^{17/6}} \frac{1}{y} \left(1 - \frac{x_1^2}{y^2}\right)\,dy = C_E\,u'^2\,L \int_{x_1}^\infty \frac{y\big(y^2 - x_1^2\big)}{\big(1 + y^2\big)^{17/6}}\,dy.
\end{equation}
To evaluate this integral analytically, we employ the algebraic substitution $u = 1 + y^2$, with $du = 2y\,dy$. Under this transformation, $y^2 - x_1^2 = u - (1 + x_1^2)$, and the lower integration limit $y = x_1$ maps to $u_0 = 1 + x_1^2$:
\begin{align}
\int_{x_1}^\infty \frac{y\big(y^2 - x_1^2\big)}{\big(1 + y^2\big)^{17/6}}\,dy 
&= \frac{1}{2} \int_{u_0}^\infty \frac{u - u_0}{u^{17/6}}\,du \nonumber \\
&= \frac{1}{2} \left[ \int_{u_0}^\infty u^{-11/6}\,du - u_0 \int_{u_0}^\infty u^{-17/6}\,du \right] \nonumber \\
&= \frac{1}{2} \left[ \frac{u_0^{-5/6}}{5/6} - u_0 \frac{u_0^{-11/6}}{11/6} \right] \nonumber \\
&= \frac{1}{2} \left( \frac{6}{5} - \frac{6}{11} \right) u_0^{-5/6} = \frac{1}{2} \left( \frac{36}{55} \right) u_0^{-5/6} = \frac{18}{55}\,\big(1 + x_1^2\big)^{-5/6}.
\end{align}
Thus, the 1D longitudinal spectrum reduces to the exact algebraic power-law form:
\begin{equation}
E_{11}(k_1) = C_{11}\,u'^2\,L\,\frac{1}{\big(1 + (k_1 L)^2\big)^{5/6}},
\label{eq:app_E11_form}
\end{equation}
where $C_{11} = \frac{18}{55} C_E$.

The amplitude $C_{11}$ is uniquely fixed by the fundamental normalization condition $u'^2 f(0) = \int_0^\infty E_{11}(k_1)\,dk_1 = u'^2$, which requires:
\begin{equation}
C_{11} \int_0^\infty \frac{d(k_1 L)}{\big(1 + (k_1 L)^2\big)^{5/6}} = 1.
\end{equation}
Setting $x = k_1 L = \tan\theta$, the integral evaluates to:
\begin{equation}
\int_0^\infty \frac{dx}{(1 + x^2)^{5/6}} = \int_0^{\pi/2} \cos^{-1/3}\theta\,d\theta = \frac{1}{2}\,\mathrm{B}\left(\frac{1}{2}, \frac{1}{3}\right) = \frac{\sqrt{\pi}\,\Gamma(1/3)}{2\,\Gamma(5/6)},
\end{equation}
where $\mathrm{B}(p, q) = \Gamma(p)\Gamma(q)/\Gamma(p+q)$ is the Euler beta function. Therefore, the exact normalized 1D longitudinal spectrum is:
\begin{equation}
E_{11}(k_1) = \frac{2\,\Gamma(5/6)}{\sqrt{\pi}\,\Gamma(1/3)}\,\frac{u'^2\,L}{\big(1 + (k_1 L)^2\big)^{5/6}}.
\label{eq:app_E11_exact}
\end{equation}

\subsection{Fourier Cosine Transform via Basset's Integral Representation}
The longitudinal autocorrelation function $f(r)$ is the Fourier cosine transform of $E_{11}(k_1)$ (Eq.~\eqref{eq:f_cos_E11}):
\begin{equation}
u'^2 f(r) = \int_0^\infty E_{11}(k_1)\,\cos(k_1 r)\,dk_1 = \frac{2\,\Gamma(5/6)}{\sqrt{\pi}\,\Gamma(1/3)}\,u'^2\,L \int_0^\infty \frac{\cos(k_1 r)}{\big(1 + (k_1 L)^2\big)^{5/6}}\,dk_1.
\end{equation}
Defining the non-dimensional spatial separation $\xi = r/L$ and integration variable $y = k_1 L$, we write:
\begin{equation}
f(r) = \frac{2\,\Gamma(5/6)}{\sqrt{\pi}\,\Gamma(1/3)} \int_0^\infty \frac{\cos(\xi y)}{\big(1 + y^2\big)^{5/6}}\,dy.
\label{eq:app_f_integral}
\end{equation}
The integral in Eq.~\eqref{eq:app_f_integral} is evaluated using Basset's classical formula for the modified Bessel function of the second kind (Macdonald function $K_\nu$) \cite{Basset1888, GradshteynRyzhik}:
\begin{equation}
\int_0^\infty \frac{\cos(a x)}{(1 + x^2)^{\nu + 1/2}}\,dx = \frac{\sqrt{\pi}}{\Gamma(\nu + 1/2)} \left(\frac{a}{2}\right)^\nu K_\nu(a) \quad \text{for } a > 0, \; \operatorname{Re}(\nu) > -1/2.
\label{eq:basset_formula}
\end{equation}
Matching exponents:
\begin{equation}
\nu + \frac{1}{2} = \frac{5}{6} \implies \nu = \frac{1}{3}, \quad a = \xi = \frac{r}{L}.
\end{equation}
Applying Eq.~\eqref{eq:basset_formula} directly yields:
\begin{equation}
\int_0^\infty \frac{\cos(\xi y)}{\big(1 + y^2\big)^{5/6}}\,dy = \frac{\sqrt{\pi}}{\Gamma(5/6)} \left(\frac{\xi}{2}\right)^{1/3} K_{1/3}(\xi).
\end{equation}
Substituting this back into Eq.~\eqref{eq:app_f_integral}:
\begin{align}
\label{eq:app_f_integral_solved}
f(r) &= \left[ \frac{2\,\Gamma(5/6)}{\sqrt{\pi}\,\Gamma(1/3)} \right] \left[ \frac{\sqrt{\pi}}{\Gamma(5/6)} \left(\frac{\xi}{2}\right)^{1/3} K_{1/3}(\xi) \right] \nonumber \\
&= \frac{2 \cdot 2^{-1/3}}{\Gamma(1/3)}\,\xi^{1/3}\,K_{1/3}(\xi) \nonumber \\
&= \frac{2^{2/3}}{\Gamma(1/3)}\,\left(\frac{r}{L}\right)^{1/3} K_{1/3}\left(\frac{r}{L}\right).
\end{align}
This completes the derivation of Eq.~\eqref{eq:vk_exact_bessel}.

\subsection{Asymptotic Behavior and Proof of Strict Positivity}
From the integral representation of the Macdonald function $K_\nu(z) = \int_0^\infty e^{-z\cosh t}\cosh(\nu t)\,dt$, $K_\nu(z) > 0$ strictly for all real $z > 0$ and $\nu \in \mathbb{R}$. Thus:
\begin{enumerate}
\item \textbf{Pointwise Positivity:} $f(r) > 0$ for all $r \in [0, \infty)$.
\item \textbf{Near-Origin Kolmogorov Scaling ($r/L \ll 1$):} Using the small-argument series expansion:
\begin{equation}
K_{1/3}(z) = \frac{\Gamma(1/3)}{2^{2/3}}\,z^{-1/3} - \frac{\Gamma(2/3)}{2^{4/3}}\,z^{1/3} + \mathcal{O}(z^{5/3}),
\end{equation}
we obtain:
\begin{equation}
f(r) = 1 - \frac{\Gamma(2/3)}{2^{2/3}\,\Gamma(1/3)}\,\left(\frac{r}{L}\right)^{2/3} + \mathcal{O}\left((r/L)^2\right),
\end{equation}
recovering $f(0) = 1$ and the classic Kolmogorov two-thirds structure function scaling $D_{LL}(r) = 2 u'^2 [1 - f(r)] \propto r^{2/3}$ in the inertial subrange.
\item \textbf{Far-Field Exponential Decay ($r/L \gg 1$):} Using the asymptotic expansion $K_\nu(z) \sim \sqrt{\frac{\pi}{2z}} e^{-z} [1 + \mathcal{O}(1/z)]$, we recover the non-oscillatory asymptotic exponential tail:
\begin{equation}
f(r) \sim \frac{2^{2/3}\sqrt{\pi}}{\sqrt{2}\,\Gamma(1/3)}\,\left(\frac{r}{L}\right)^{-1/6}\,\exp\left(-\frac{r}{L}\right) \quad \text{as } r/L \to \infty,
\end{equation}
confirming strictly positive, non-oscillating decay without negative loops out to infinity.
\end{enumerate}

\subsection{The von K\'arm\'an--Pao Spectrum with Viscous Cutoff: Asymptotics and Numerical Investigation}

When the exponential viscous cutoff of Pao \cite{Pao1965} is incorporated, the full three-dimensional energy spectrum becomes:
\begin{equation}
E_{\text{VKP}}(k) = C_E\,u'^2\,L\,\frac{(kL)^4}{\big(1 + (kL)^2\big)^{17/6}}\,\exp\left(-\frac{3}{2} C_K (k\eta)^{4/3}\right),
\label{eq:app_vkp_3d}
\end{equation}
where $\eta = (\nu^3/\epsilon)^{1/4}$ is the Kolmogorov microscale and $C_K \approx 1.5$ is the Kolmogorov constant.

\paragraph{Absence of Closed-Form Special Functions.}
Because the Pao viscous damping factor $\exp(-\frac{3}{2} C_K (k\eta)^{4/3})$ features a fractional exponent $4/3$, the composite integrand in the longitudinal projection:
\begin{equation}
u'^2 f_{\text{VKP}}(r) = \int_0^\infty E_{\text{VKP}}(k)\,K_{11}(kr)\,dk = 2 \int_0^\infty E_{\text{VKP}}(k)\,\frac{\sin(kr) - kr\cos(kr)}{(kr)^3}\,dk
\end{equation}
combines a power-law factor with a stretched-exponential decay. Consequently, unlike the pure von K\'arm\'an spectrum, $f_{\text{VKP}}(r)$ cannot be expressed in closed form to the best of our knowledge.

\paragraph{Asymptotic Decomposition and Multiscale Sign Behavior.}
The sign behavior of $f_{\text{VKP}}(r)$ across the entire domain $r \in [0, \infty)$ is established by matching three physical regimes:
\begin{enumerate}
\item \textbf{Dissipative Subrange ($r \lesssim \eta$):} Viscous diffusion eliminates the singular $(r/L)^{2/3}$ inertial gradient at the origin, regularizing the correlation into an analytic parabolic expansion:
\begin{equation}
f_{\text{VKP}}(r) = 1 - \frac{r^2}{2\lambda_g^2} + \mathcal{O}(r^4) > 0,
\end{equation}
where $\lambda_g = \left( 15\nu u'^2/\epsilon \right)^{1/2}$ is the transverse Taylor microscale.
\item \textbf{Inertial and Integral Subranges ($\eta \ll r \lesssim L$):} Because $(k\eta)^{4/3} \ll 1$ throughout the energy-containing wavenumber range $k \lesssim 1/L$, the Pao damping factor is exponentially close to unity ($\exp(-c(k\eta)^{4/3}) \approx 1$). Thus, $f_{\text{VKP}}(r)$ converges to the exact closed-form Bessel solution derived in Eq.~\eqref{eq:app_f_integral_solved}:
\begin{equation}
f_{\text{VKP}}(r) \approx \frac{2^{2/3}}{\Gamma(1/3)}\left(\frac{r}{L}\right)^{1/3} K_{1/3}\left(\frac{r}{L}\right) > 0.
\end{equation}
\item \textbf{Far-Field Decay ($r \gg L$):} Because viscous dissipation is negligible at large scales ($k \sim 1/L \ll 1/\eta$), the far-field behavior is governed by the large-argument asymptote of the Bessel solution $K_{1/3}(r/L) \sim \sqrt{\pi / (2r/L)} e^{-r/L}$, preserving the non-oscillatory exponential decay $f_{\text{VKP}}(r) \propto (r/L)^{-1/6} e^{-r/L} > 0$.
\end{enumerate}

\paragraph{Arbitrary-Precision Quadrature Verification via \texttt{mpmath}.}
To verify that $f_{\text{VKP}}(r)$ remains strictly positive across all scales without relying on asymptotic matching alone, we evaluate the continuous integral representation using the arbitrary-precision Python library \texttt{mpmath} (setting working precision to 25 decimal digits, \texttt{mp.dps = 25}):
\begin{equation}
u'^2 f_{\text{VKP}}(r) = \int_0^\infty E_{\text{VKP}}(k) \left[ \frac{2(\sin(kr) - kr\cos(kr))}{(kr)^3} \right] dk.
\end{equation}
For $k r \ll 1$, Taylor expansion of the kernel $K_{11}(x) = \frac{2}{3} - \frac{1}{15}x^2 + \frac{1}{420}x^4 - \frac{1}{22680}x^6 + \dots$ avoids catastrophic numerical cancellation near $k=0$, while adaptive quadrature evaluates the semi-infinite integral to arbitrary precision. High-precision numerical evaluation yields:
\begin{align}
f_{\text{VKP}}(r/L = 0.001) &= 0.99758333\dots > 0, \nonumber \\
f_{\text{VKP}}(r/L = 1.0) &= 0.26297558\dots > 0, \nonumber \\
f_{\text{VKP}}(r/L = 5.0) &= 0.00382475\dots > 0, \nonumber \\
f_{\text{VKP}}(r/L = 10.0) &= 2.31004\times 10^{-5} > 0, \nonumber \\
f_{\text{VKP}}(r/L = 30.0) &= 3.66904\times 10^{-13} > 0.
\end{align}
This numerical verification demonstrates that $f_{\text{VKP}}(r)$ is strictly positive across more than 13 decades of amplitude, confirming that the negative dips produced by discrete Fourier transforms are purely numerical artifacts of finite-domain discretization and spectral aliasing.

\bibliographystyle{apsrev4-2}
\bibliography{references}

\end{document}